\documentclass[10pt,a4paper]{amsart}
\usepackage{amstext, amssymb, amsthm, amsfonts, latexsym,bbm,tikz,sfmath,color,varioref}
\usetikzlibrary{matrix,positioning}

\usepackage[version=4]{mhchem}

\usepackage{enumerate}
\usepackage[all]{xy}
\usepackage{graphicx}
\usepackage{pifont,mathtools}

\usepackage{hyperref}
\usepackage{float}

\usepackage{setspace}
\usepackage{dsfont}
\usepackage{geometry}
\usepackage{enumitem}
\usepackage{todonotes}
\usepackage{bbm}
\usepackage[utf8]{inputenc}
\usepackage[english]{babel}
\usepackage{textcomp}

 \usepackage{algorithm}
\usepackage{algpseudocode}
\usepackage{comment}
\usepackage{tikz-cd}
\usepackage{hyperref}

\def\cR{\mathcal{R}}
\def\cS{\mathcal{S}}
\def\cC{\mathcal{C}}
\def\cH{\mathcal{H}}

\makeatletter
\newcommand\level[1]{%
  \ifcase#1\relax\expandafter\chapter\or
    \expandafter\section\or
    \expandafter\subsection\or
    \expandafter\subsubsection\else
    \def\next{\@level{#1}}\expandafter\next
  \fi}
\newcommand{\@level}[1]{%
  \@startsection{level#1}
    {#1}
    {\z@}%
    {-3.25ex\@plus -1ex \@minus -.2ex}%
    {1.5ex \@plus .2ex}%
    {\normalfont\normalsize\bfseries}}

\newcounter{level4}[subsubsection]
\@namedef{thelevel4}{\thesubsubsection.\arabic{level4}}
\@namedef{level4mark}#1{}
\count@=4
\loop\ifnum\count@<100
  \begingroup\edef\x{\endgroup
    \noexpand\newcounter{level\number\numexpr\count@+1\relax}[level\number\count@]
    \noexpand\@namedef{thelevel\number\numexpr\count@+1\relax}{%
      \noexpand\@nameuse{thelevel\number\count@}.\noexpand\arabic{level\number\numexpr\count@+1\relax}}
    \noexpand\@namedef{level\number\numexpr\count@+1\relax mark}####1{}}
  \x
  \advance\count@\@ne
\repeat
\makeatother
\newtheorem{theorem}{Theorem}[section]
\newtheorem{lemma}[theorem]{Lemma}
\newtheorem{corollary}[theorem]{Corollary}
\newtheorem{proposition}[theorem]{Proposition}

\theoremstyle{definition}
\newtheorem{definition}{Definition}[section]

\newtheorem{example}[theorem]{Example}

\theoremstyle{remark}
\newtheorem{remark}[theorem]{Remark}

\def\la{\lambda}

\def\1{\mathbf{1}}
\def\la{\lambda}

\def\cR{\mathcal{R}}
\def\cS{\mathcal{S}}
\def\cC{\mathcal{C}}

\def\E{\mathbb{E}}
\def\R{\mathbb{R}}
\def\N{\mathbb{N}}
\def\Z{\mathbb{Z}}

\numberwithin{equation}{section}

\begin{document}

\title[On the Abelian Structure of Noncompetitive Chemical Reaction Networks]{\vspace*{-1.4cm}On the Abelian Structure of Noncompetitive Chemical Reaction Networks}

\author{
Louis Faul\textsuperscript{1}, 
Xavier Richard\textsuperscript{2}, 
Marie B\'etrisey\textsuperscript{1},
Christian Mazza\textsuperscript{1}
}

\thanks{\textsuperscript{1} Department of Mathematics, University of Fribourg, Fribourg, Switzerland}
\thanks{\textsuperscript{2} University of Applied Sciences of Western
Switzerland (HES-SO), Sion, Switzerland}

\thanks{\textbf{Corresponding author:} Christian Mazza, Email: \href{mailto:christian.mazza@unifr.ch}{christian.mazza@unifr.ch}}

\date{}
\begin{abstract}
Chemical reaction networks (CRNs) are foundational models for describing complex biochemical processes.
We study noncompetitive CRNs, a class of networks whose static states—where the CRN is inactive—are rate-independent, and that can implement ReLU neural networks. 
CRNs of interest in biochemistry and systems biology are embedded in complex networks so that  CRNs have to respond to internal and environmental cues. We describe the network's response to such perturbations  using a new Markov chain that we call CRN sandpile Markov chain, whose state space is the set of static states. The transition mechanism of the CRN sandpile Markov chain is   defined by adding a molecule of a randomly chosen species  to  a static state,  and then letting the CRN state evolve toward a new static state.
A central contribution of the present  work is  the observation that 
one can associate a natural Abelian Network (AN) to each noncompetitive CRN, and use AN theory to get new mathematical results on noncompetitive CRNs.
For  noncompetitive CRNs on a finite state space, we use AN theory to get that only a fraction of the static states are recurrent for the CRN sandpile Markov chain. We obtain furthermore that the set of recurrent states   is in one to one correspondence with the critical group of the AN, which plays a major role in AN theory.
 Overall, this work establishes a unified algebraic and probabilistic framework for analyzing the long-term behavior of noncompetitive CRNs.  We focus on a special class of noncompetitive CRNs called generalized toppling networks, and obtain new mathematical results both for the CRN and AN settings.
 \end{abstract}

\clearpage\maketitle
\thispagestyle{empty}
\vspace*{-0.6cm}

\color{black}
\section{Introduction}\label{summary_1}
Biochemical networks perform complex tasks in small volume and noisy environments, like e.g. transducing signals through signaling networks or controlling gene regulatory networks. 
Recent research in molecular programming has explored mathematical and computational models of such complex biochemical pathways that can be mapped onto architectures which are compatible with the kind of computations that cells perform in natural environments, like for example signaling cascades,  see, e.g., \cite{Hellingwerf,Samaniego}  where the related biochemical net can be expressed as feedforward phosphorylation-based perceptron.

On the other hand, natural biochemical network dynamics are modeled using chemical reaction network theory (CRN). Chemical implementation of neural network (NN) architectures as well as the associated learning processes have been proposed \cite{Hellingwerf,Bray,Hwa,manoj,AndersonNN,Samaniego}. We focus here on noncompetitive CRNs \cite{Vasic} that are rate-independent, whose equilibria are robust to reaction rates and kinetic laws, and depend solely on the stoichiometric structure of the CRN. Noncompetitive CRNs can furthermore implement ReLU NN, with continuous input variables which are seen as continuous species concentrations in the CRN implementation. To address the complex computational tasks that cell must perform in small volumes and with low molecular counts, some works have proposed to use discrete NN models like chemical Boltzmann machines \cite{manoj,Hwa}. We consider here robust discrete noncompetive CRNs.\\

CRNs of interest in biochemistry and systems biology are embedded in complex networks so that  CRNs have to respond to internal and environmental cues. 
We describe the network's response to such perturbations  using a CRN sandpile Markov chain on the set of static states where the CRN is inactive.
The transition mechanism of the Markov chain is obtained 
     by adding a molecule of a randomly chosen species  to static state, and then letting reactions occur until a new static state is reached. 
     
A central contribution of the present  work is  the observation that 
one can associate a natural Abelian Network (AN) to each
rate-independent noncompetitive CRN, and use AN theory to get new mathematical results on noncompetitive CRNs. 
AN theory is a well-established framework for self-organized criticality, see e.g. \cite{jensen1998self} for an introduction to this concept, and
     AN models are generalizations of  abelian sandpile models (ASM) that were originally introduced in physics as  models for avalanches \cite{Bak,Dhar}. Such models have since been extensively studied in mathematics  and  physics
 \cite{Biggs2,meester2003abelian}.

 For  noncompetitive CRNs of finite state space, we use AN theory to
 study the set of recurrent states of the CRN sandpile  Markov chain. 
 Such states
      will appear infinitely often along the orbits of the dynamics, and it seems plausible that such states play some particular biological role,  while the other static states will not persist in the long run.
We get that the fraction of static states that are recurrent  is in one to one correspondence with the critical group of the AN, leading to exact formulas for the number of recurrent states.

For the purpose of an example, consider the following network of reactions among species of $\mathcal{S}=\{1,2\}$

 \begin{equation}\label{e.intro}   
    2X_1 \longrightarrow X_2 , \ \
    2X_2 \longrightarrow X_1,
\end{equation}
where $X_i$ is a symbol for species $i$, $i=1,2$.
In the discrete setting, CRN states $c=(c_1,c_2)\in\N^2$ are such that $c_i\in\N$ gives the number of molecules of species $i$, $i=1,2$ that are present in state $c$.
The CRN sandpile Markov chain adds a molecule from a randomly chosen species to static state $q$ and lets reactions occur until a new static state $q'$ is reached. 
The static states are 
$q_1 = (0,0),q_2 = (0,1),q_3 = (1,0)$ and $q_4 =(1,1)$. 
Figure \ref{f.intro_example} shows that  $q_2$, $q_3$ and $q_4$ are recurrent. We observe that once the Markov chain has reached one of the recurrent states (shaded brown), it will never reach $q_1$ again. This result will be proved later in Example \ref{exemple1}.
\begin{figure}[h!]
    \centering
    \includegraphics[width=0.3\linewidth]{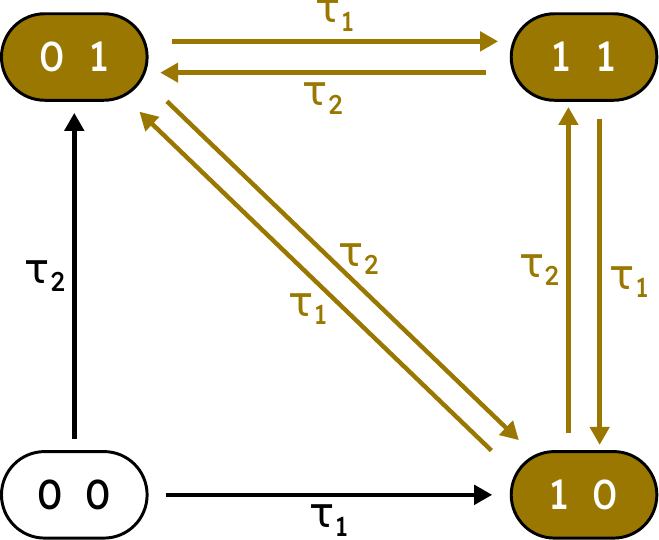}
    \caption{Transitions of the CRN sandpile Markov chain associated to the CRN: $2X_1 \longrightarrow X_2,  2X_2 \longrightarrow X_1$. Recurrent states are shaded brown.The operation of adding one molecule of $X_1$ (resp. $X_2$) and letting the system stabilize is denoted $\tau_1$ (resp. $\tau_2$)} 
    \label{f.intro_example}
\end{figure}

     We further identify structural conditions under which noncompetitive CRNs are execution bounded, meaning that the number of   reactions needed to move $q$ to $q'$ is finite.
     We show that these conditions are equivalent to the halting property of the associated AN. AN theory considers production matrices $P$ that are defined algebraically. Basically, the $(i,j)^{th}$ entry of the production matrix represents the average number of molecules of species $i$ created while processing one molecule of species $j$.
    The production matrix associated with the noncompetitive CRN (\ref{e.intro}) is  given by
 (see  Proposition \ref{p.production} for more details)   
$$P = \begin{pmatrix}
    0 & 1/2 \\
    1/2 & 0
\end{pmatrix}.$$

We show  how a special class of CRNs called generalized toppling networks   can be studied 
using results from AN  theory.
A central result given by Lemma \ref{l.halting}  states that  CRNs that correspond to  generalized toppling networks are execution bounded if and only if the spectral radius of the production matrix
of the associated AN is strictly smaller than one, a simple example being given by the CRN (\ref{e.intro}).
One then obtains a formula giving the number of recurrent states of the CRN sandpile Markov chain as the determinant of a reduced Laplacian matrix from AN theory. In the other direction, results from CRN theory related to execution boundedness, based on the existence of linear potential functions, lead to new results within the AN framework.
Overall, this work establishes a unified algebraic and probabilistic framework for analyzing the long-term behavior of noncompetitive CRNs.

Section \ref{Preliminaries} introduces the relevant notions from CRN theory. Section \ref{s.noncompetitive} provides basic results on noncompetitive CRNs by specifying e.g. the nature of the graph of complexes and recalling that ReLU NN can be embedded in the framework of noncompetitive CRNs. It is shown that the well known Abelian sandpile model (ASM) can be seen as noncompetitive CRN. 
 Section \ref{s.markov_chains} introduces the notion of CRN sandpile Markov chain for noncompetitive CRNs.
Section \ref{s.AbelianNet} presents basic definitions and notions from Abelian network theory.
Section \ref{s.NonCompetAbelianNet} shows how noncompetitive CRNs can be seen as AN. 
Section \ref{s.sandpile-generalized} provides new mathematical results linking CRN and AN sandpiles Markov chains for generalized toppling networks. Lemma \ref{l.map} relates CRN and AN executions, and shows that CRN and AN sandpile Markov chains are identical, while Lemma \ref{l.halting} shows that generalized toppling networks are execution bounded if and only if the spectral radius of the production matrix is smaller than 1. 
Section \ref{s.examples} gives examples of noncompetitive CRNs and shows how our results permit to find recurrent states. Section \ref{s.discussion} discusses the limitations and possible extensions of our work. We also provide an Appendix in Section \ref{BasicsAbelianNetwork} giving the main notions and mathematical results from AN theory.
Figure \ref{ModelsStructure} shows the hierarchical structure of the considered models.

    \begin{figure}[h]
    \centering
    \includegraphics [width=7cm]{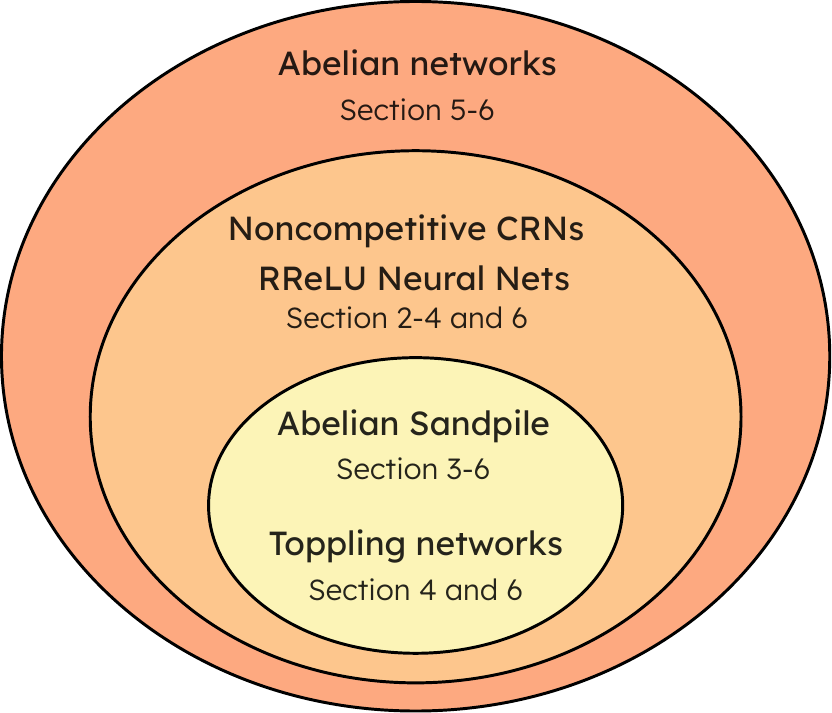}
    \caption{Graphical representation indicating the hierarchy among models\label{ModelsStructure}}
\end{figure}

\section{Preliminaries on CRNs \label{Preliminaries}
}
\subsection{Notations}
For any integers $m$ and $n$, we denote by $\Z^m$ (respectively  $\N^m$) the m-dimensional lattice with (respectively  non-negative) integer entries. Similarly, $\mathbb{R}_{\geq 0}^m$ (respectively $\mathbb{R}_{> 0}^m$) denotes the set of $m$-dimensional real vectors with non-negative (respectively positive) coordinates. The set of $n \times m$ matrices with real entries is denoted by $\mathbb{R}^{n\times m}$, and we denote by $I_n$ the identity matrix of size $n$. For a matrix $M$ in $\mathbb{R}^{n\times m}$ and integers $i,j \in \{1,...,n\}\times \{1,...,m\}$, we denote by $M_{i,.} \in \mathbb{R}^m $ (resp. $M_{.,j} \in \mathbb{R}^n$) the $i^{th}$ line (resp. $j^{th}$ column) of $M$. For two vectors $a,b \in\mathbb{R}^n$, we write $a \leq b$ if $a$ is less or equal to $b$ component-wise i.e. if $ \forall i \in \{1,...n\}, a_i \leq b_i$. Similarly, we use the notations $a < b$, $a\geq b$ and $a>b$. We denote by $0_n$ the vector of $\mathbb{R}^{n}$ with all components being $0$. 
For $k \in \{1,...,n\}$, we denote by $e_k \in \R^n$, the unit vector with a one in the $k^{th}$ component and $0$ everywhere else. Let $A$ be a subset of $\{1,..,n\}$, we denote $1_A = e_A = \sum\limits_{k\in A} e_k$. 

\subsection{Chemical Reaction Networks}\label{CRN}
In this subsection, we provide a brief introduction to the formalism of chemical reaction networks.  For a more detailed study, we refer to \cite{feinberg2019foundations}.
\begin{definition}
    A \emph{Chemical Reaction Network }(CRN) is a triple $\cH=(\cS,\cC,\cR)$ where $\cS$ is the set of \emph{species}, $\cS=\{1,\cdots,N\}$, $\cC$ is the set of \emph{complexes} and $\cR$ is the set of \emph{reactions}, $\cR=\{R_1,\cdots,R_d\}$. 
\end{definition}

The set of \emph{species} is regarded as an orthogonal basis of the Euclidean space $\R^N$, which means that each species $i$ corresponds to the canonical basis vector $e_i$ of $\R^N$.
 A \emph{complex} $\nu$ is a non-negative integer linear combination of species, represented as a vector in $\N^N$, where each component specifies the number of molecules of the corresponding species.
 
The \emph{reaction vector} associated with the $r^{th}$ reaction $\nu_r \to \nu'_r$ where $\nu_r$ is called the \emph{source complex} and $\nu'_r$ the \emph{product complex}, is the integer-valued vector $l_r\in\Z^N$, defined by
 $$l_r =\nu'_r -\nu_r,$$
 This vector indicates, for each species, the net number of molecules of consumed (reactants) or created (products) during the $r^{th}$ reaction. We define the \emph{reactant matrix} $\Gamma^r \in \mathbb{R}^{N \times d}$ and the \emph{product matrix} $\Gamma^p \in \mathbb{R}^{N \times d}$ as follows: $\Gamma^r_{ir}$ denotes the number of molecules of species $i$ appearing as reactants (on the left-hand side) in reaction $r$, while $\Gamma^p_{ir}$ denotes the number of molecules of species $i$ appearing as products (on the right-hand side) in reaction $r$. Notice that $\nu_r = \Gamma^r_{.,r}$ and $\nu'_r = \Gamma^p_{.,r}$.
 The \emph{stoichiometric matrix} $\Gamma \in \mathbb{R}^{N \times d}$ is defined as the difference between the product matrix and the reactant matrix, and has the reaction vectors as its columns, that is,
\[
\Gamma = \Gamma^p - \Gamma^r =  [l_1 \;\; l_2 \;\; \dots \;\; l_d],
\]
where for $r \in \{1,...,d\}$, $l_r = \nu_r' - \nu_r$ corresponds to the $r^{\text{th}}$ reaction vector.
\begin{definition}\label{r.sink}
For a given CRN, a species that appears in some product complex but not
in any source complex is called a sink species.
\end{definition}

\begin{definition}
    A complex $\nu \in \mathcal{C}$ is called \emph{single-species} or \emph{multi-species} if it contains one species or more than one species, respectively.
\end{definition}

 The next section introduces basic notions related to the graphical description of the transition graph of a CRN.
 \subsection{The graph of complexes}\label{s.graph_of_complexes}
 
 The set ${\mathcal C}$ of chemical complexes\index{chemical complex}
 associated with ${\mathcal S}$ and ${\mathcal R}$ is defined as ${\mathcal C}=\cup_{r \in \mathcal{R}} \{\nu_r,\nu'_r\}$. We define ${\mathcal C}_S$ and ${\mathcal C}_P$ to be the sets of source and product complexes, respectively. Note that ${\mathcal C}_S \cup {\mathcal C}_P = {\mathcal C}$, although they are not necessarily disjoint. The set of reactions $ \{\nu_r \rightarrow\nu'_r;\ r\in {\mathcal R}\}$ can then be described using directed edges by writing ${\mathcal E}= \{(\nu_r ,\nu'_r);\ r\in {\mathcal R}\}$. This defines a {\bf directed graph}  ${\mathcal G}=({\mathcal C},{\mathcal E})$,  with node set ${\mathcal C}$ and edge set ${\mathcal E}$. Let $m= \vert {\mathcal C}\vert$ and $d=\vert {\mathcal E}\vert$. 
The support of a complex $\nu= \sum_{s=1}^N \lambda_s e_s$ is defined by
$${\rm supp}(\nu)= \{s\in {\mathcal S}:\ \lambda_s \ne 0\}.$$

 \begin{figure}[h!]
\centering 
   \includegraphics[width=0.6\linewidth]{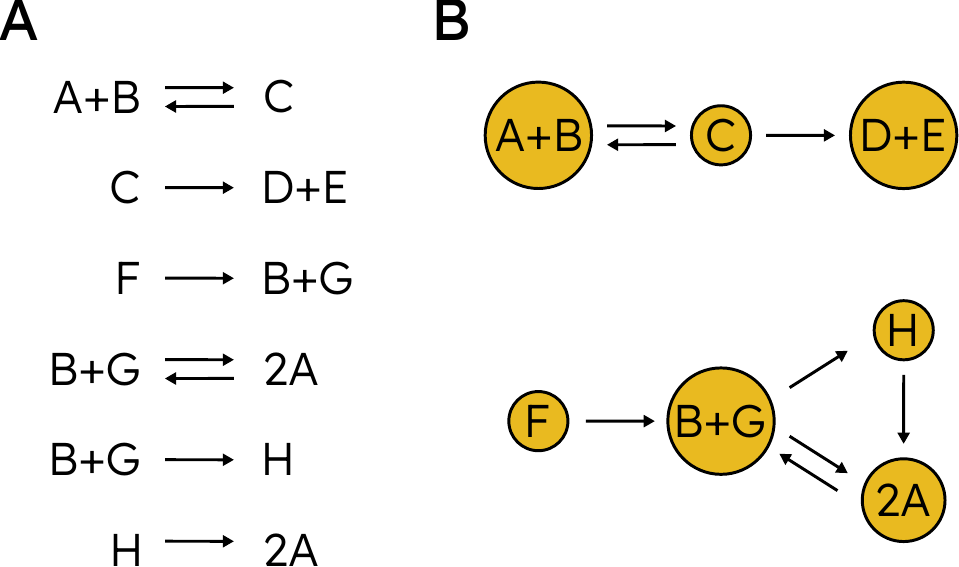}
   \caption{(A) A chemical reaction network. (B) The corresponding graph of complexes.\label{F1}}
\end{figure}

 \subsection{CRN executions and static states}\label{s.static-state}

A {\bf state} of a CRN is a vector $c \in \mathbb{N}^N$. 
A reaction $r$ given by $\nu_r \to \nu'_r$ is said to be applicable in state $c$ when $\nu_r \le c$.
If reaction $r$ is applicable at $c$, its application leads to the new state $c'=c+\nu'_r -\nu_r$, and one writes
$c \to  c'$.

\begin{definition}\label{ExecutionCRN}
    A \emph{CRN} execution $\mathcal{F}$ is a finite or infinite sequence of states $\mathcal{F} = (c^0,c^1,...)$, where $c^n \in \mathbb{N}^N$, such that $c^{n-1} \to  c^n$ and $c^{n-1} \neq c^n$ for all $n$. A state $c' \in \mathbb{N}^N$ is said to be {\bf reachable} from another state $c \in \mathbb{N}^N$, if there exists a finite CRN execution  $\mathcal{F} = (c^0,c^1,\ldots,c^n)$ of length $\ell_{\mathcal{F}}=n$ with $c^0=c$ and $c^n =c'$, and we write $c\hookrightarrow c' $.
\end{definition}

\begin{definition}\label{static}
Given a state $c \in \mathbb{N}^N$; the CRN is active at $c$ as long as $c \geq \nu_r$ for some $r \in \{1,...,d\}$, otherwise we say that the CRN is inactive at $c$.
 A state $c \in \mathbb{N}^N$ at which the CRN is inactive is called a {\bf static state}. We denote by $\operatorname{Stat}\cH$ the set of static states of a CRN $\cH$.
 \end{definition}

 \subsection{Stoichiometric subspace}
 
 Observing that chemical reactions are obtained by adding reaction vectors $l_r$, $r=1,\cdots,d$, one naturally defines
 the \emph{stoichiometric subspace}\index{stoichiometric subspace} ${\mathcal T}$  as the span of the reaction vectors $l_r=\nu'_r -\nu_r$.  For $x\in\R^N_{\geq 0}$ 
 $x+{\mathcal T}$ is called a translation of the stoichiometric space, and
  the sets $(x+{\mathcal T})\cap \R_{\geq 0}^N$ and $(x+{\mathcal T})\cap \R_{> 0}^N$ are referred to as
 the \emph{stoichiometric compatibility classes}  and the \emph{positive stoichiometric compatibility classes} associated with $x$, respectively. Similarly, we define the \emph{discrete stoichiometric subspace} ${\mathcal T'}$ as the set of all the linear combinations of the reaction vectors with integer coefficients, and for $x\in\Z^N_{\geq 0}$, $(x+{\mathcal T'})\cap \Z_{\geq 0}^N$ and $(x+{\mathcal T'})\cap \Z_{> 0}^N$  are referred as the \emph{discrete stoichiometric compatibility classes} and the \emph{discrete positive stoichiometric compatibility classes} associated with x, respectively.


\section{Non competitive CRNs\label{s.noncompetitive}}

We focus on a particular class of chemical reaction networks known as noncompetitive CRNs \cite{Vasic}. This class exhibits favorable properties both in terms of convergence and in their ability to approximate neural networks. These models are motivated by the experimentally observed computational robustness of biochemical CRNs, that are embedded in complex and fluctuating biochemical environments in systems biology.

\begin{definition}\label{NonCompetitive}
A CRN is said to be {\bf noncompetitive} if every species whose molecular count is decreased in a reaction appears as a reactant in that reaction only.
\end{definition}

\begin{remark}
    We will see in Section \ref{s.markov_chains} that static states are absorbing for the usual CTMC associated to CRNs but not for the new  CRN sandpile Markov chain that we will introduce. 
\end{remark}
\begin{remark}\label{r.unique-static}
     Noncompetitive CRNs have the property that each  stoichiometric compatibility class contains exactly one static state for continuous and discrete versions of the CRNs
  
  \end{remark}
  
 \begin{theorem}\label{ThVasic}[Theorem 4 of \cite{Vasic}]
  Assume a   CRN is noncompetitive. If
$c  \hookrightarrow c' $ via a CRN execution $\mathcal{F}$ with length $\ell_{\mathcal{F}}$ and $c'$  is static, then there is no
path from $c$ with length longer than $\ell_{\mathcal{F}}$, any path with length $\ell_{\mathcal{F}}$
also ends in $c'$, and any path with length shorter than $\ell_{\mathcal{F}}$ ends in a
state which is not static.   
 \end{theorem}


\subsection{Execution boundedness}


We describe the conditions under which a  CRN execution is finite.

\begin{definition}
    A CRN is said to be \emph{execution bounded} at $c \in \mathbb{N}^N$ if all CRN executions starting at $c$ are finite, and a CRN is \emph{entirely execution bounded} if it is execution bounded from every initial state.
\end{definition}
Intuitively, a CRN will be execution bounded if the total "energy" of the system decreases with each reaction. We now introduce some tools and give an explicit criterion for deciding whether a CRN is execution bounded.

\begin{definition}\label{d.linear_potential}
    A \emph{linear potential function} $\psi: \mathbb{R}^N_{\geq 0} \to \mathbb{R}_{\geq 0}$ is a non-negative linear function such that, for each reaction $ \nu \to\nu'\in \mathcal{E}$, $\psi(\nu')-\psi(\nu) <0$.
\end{definition}
By definition, 
$$\psi(a) = \sum_{i= 1}^N v_i a_i \geq 0,\ \forall  a \in \mathbb{R}^N_{\geq 0},$$
and therefore $v_i \geq 0$ for all $i\in \{1,...,N\}$.

\begin{theorem}(\cite{doty2024computational},Theorem 6.5, and \cite{czerner})\label{condition_bounded}
    A CRN has a linear potential function if and only if it is entirely execution bounded.
\end{theorem}
To decide whether a CRN has a linear potential function, one has to determine whether a given system of linear inequalities has a solution using linear programming (see \cite{doty2024computational}).


\subsection{Structural description of graphs of complexes for noncompetitive CRNs\label{s.structural}}


     Let $U_0$ be the set of species that never decrease along any reaction, such as species that are catalysts or enzymes.  The noncompetitivity of the CRN then implies that one can partition the set ${\mathcal{S}}\setminus U_0$ which are decreased in some reaction into a disjoint union of subsets $U_k\subset {\mathcal S}$,  where $k=1,\cdots, p$, $p\le d$.

\begin{definition}\label{ComplexTypes}
Let $\nu \in\mathcal{C} $. $\nu$ is a type 0 complex when ${\rm supp}(\nu)\subset U_0$. $\nu$ is a type 1a complex when
$${\rm supp}(\nu)\cap ({\mathcal S}\setminus U_0) = U_k,$$ for a unique subset $U_k$, $k\ne 0$. $\nu$ is a complex of type 1b when
${\rm supp}(\nu)\cap ({\mathcal S}\setminus U_0)$ is strictly included in some $U_k$, $k\ne 0$.
$\nu$ is a type 2 complex when
$${\rm supp}(\nu)\cap U_k\ne \emptyset,\ \  {\rm supp}(\nu)\cap U_{k'}\ne \emptyset,
$$
for some $k\ne k'$ and $k,k'\ne 0$.

\end{definition}

\begin{lemma}\label{StructuralComplexGraph}
Assume that the CRN is noncompetitive. Complexes of types 1b and 2 cannot be source complexes.
Let $\nu\in\mathcal{C}_S$ be a type 1a complex such that $(\nu,\nu')\in\mathcal{E}$ for a unique reaction. Then $\nu$ and $\nu'$
are such that
\begin{equation}\label{Comb1}
 \nu = \sum_{i\in U_0}\alpha_i e_i + \sum_{i\in U_k} \beta_{ik} e_i,
 \end{equation}
 where $e_i$ is the canonical basis vector of $\R^N$ associated with species $i$, and
 \begin{equation}\label{Comb2}
 \nu' = \sum_{i\in U_0}\alpha'_i e_i + \sum_{l=1}^p\sum_{i\in U_l} \beta'_{il} e_i,
 \end{equation}
   for stoichiometric coefficients $\alpha_i$, $\alpha'_i$, $\alpha'_i$, $\beta_{ik}$
 and $\beta'_{il}$, that satisfy the constraints  $\alpha'_i\ge \alpha_i$, for all $i\in U_0$, and
 $\beta'_{il}<\beta_{il}$, for all $i\in U_l$, $l=1,\ldots,p$.

\end{lemma}

\begin{figure}[h!]
\centering 
   \includegraphics[width=0.8\linewidth]{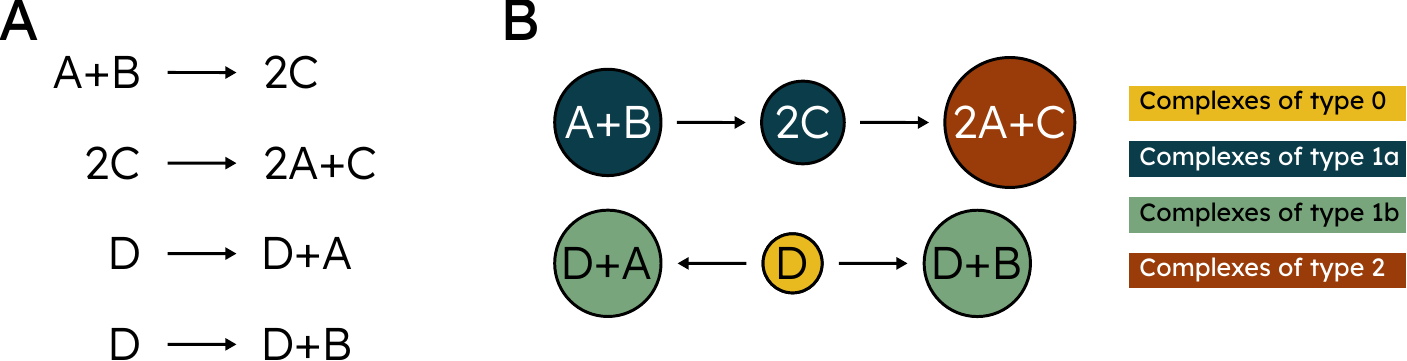}
   \caption{Example of noncompetitive CRN. A: The set of involved reactions. B: The related graph of complexes with the four possible complex types.\label{F2}}
   \end{figure}

Lemma \ref{StructuralComplexGraph} implies that the graphs of complexes of noncompetitive CRNs possess the following generic properties:
\begin{itemize}
\item{} Complexes of type 0 with ${\rm supp}(\nu)\subset U_0$ can be the origin of possibly multiple edges of ${\mathcal E}$.
\item{} Complexes of type 1a with ${\rm supp}(\nu) \cap ({\mathcal S}\setminus U_0)   = U_k$ for some $k$ can be the origin of a unique edge of ${\mathcal E}$.
\item{} Complexes of types 1b and 2 cannot be the origin of a directed edge of ${\mathcal E}$.

\end{itemize}

\begin{lemma}\label{no3}
    In an execution bounded noncompetitive CRN, complexes of type 0 cannot be the origin of a directed edge of ${\mathcal E}$.
\end{lemma}
\begin{proof}
    Suppose that $(\nu, \nu')\in \mathcal{E}$, with $\nu=\sum_{s=1}^N \lambda_s e_s$, $\nu'=\sum_{s=1}^N \lambda'_s e_s$ and ${\rm supp}(\nu)\subset U_0$. Then $\nu' \geq \nu$ with a strict inequality for at least one species. For any non-negative linear function $\psi : \mathbb{R}^N_{\geq 0} \to \mathbb{R}_{\geq 0}$, 
    with $\psi(a)=\sum_{i=1}^N v_i a_i$, $v_i\ge 0$ $\forall i$, one gets that
    $$\psi(\nu') - \psi(\nu) = \sum_{s=1}^N v_s \left(\lambda'_s-\lambda_s\right) \geq 0.  $$
    So there is no linear potential function, and the CRN cannot be execution bounded.
\end{proof}

\begin{corollary}\label{Structure}
Assume that the CRN is execution bounded and noncompetitive. $\mathcal{C}_S$ cannot contain type 0, 1b and 2 complexes.
Let $\mathcal{C}_f = {\mathcal C}_P\cap \mathcal{C}_S^c$. 
The graph of complexes is a directed pseudoforest or a functional graph for the mapping $F:\ \mathcal{C}\to\mathcal{C}$ such that
 $F(\nu)=\nu'$ for edges $(\nu\to\nu')\in\mathcal{E}$ where $\nu$ is a type 1a source complex,  and with fixed points from $\mathcal{C}_f$. Each connected component of 
 $(\mathcal{C},\mathcal{E})$ is therefore a pseudotree, that is, a directed tree rooted on either a complex of type 0, 1a, 1b or 2, or  on a cycle composed of type 1a complexes.
\end{corollary}

 \begin{figure}[h!]
\centering 
   \includegraphics[width=0.8\linewidth]{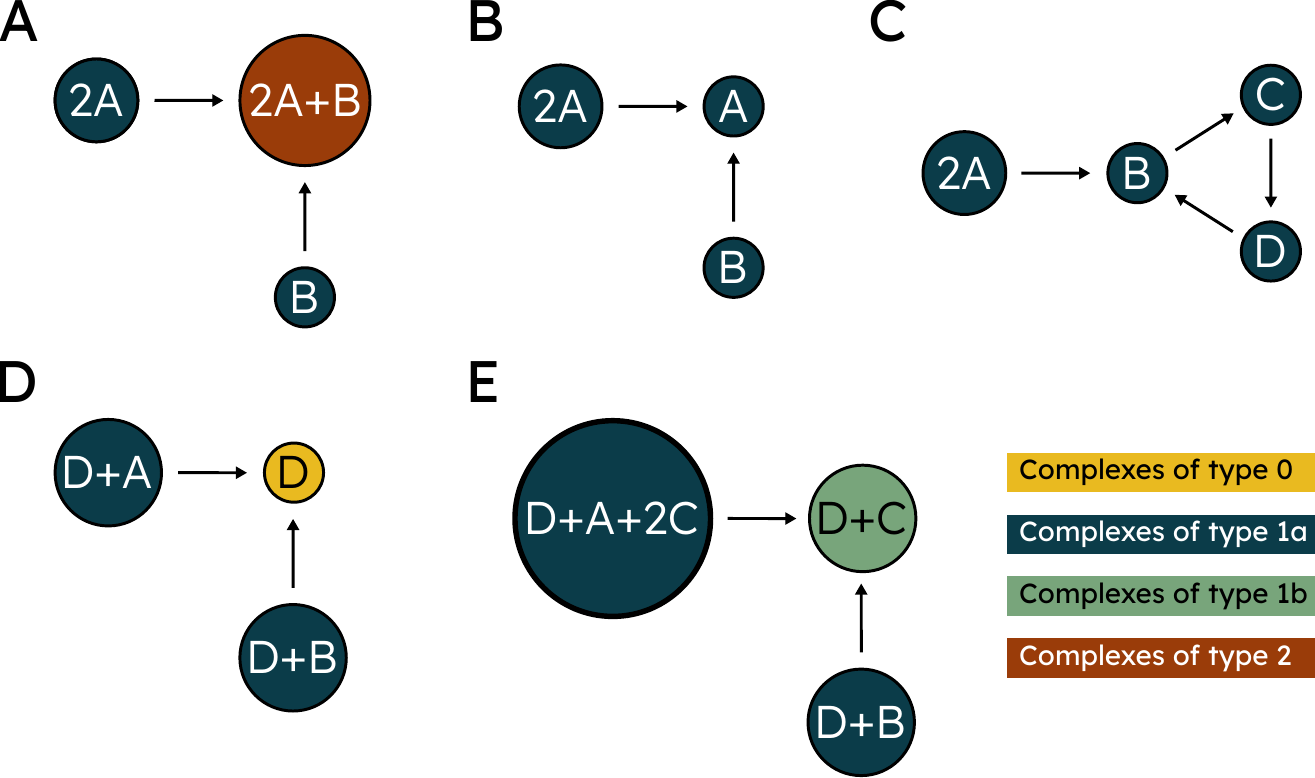}
   \caption{Examples of graph of complexes for noncompetitive execution bounded CRNs, as described in Corollary \ref{Structure}.A: Directed tree rooted at a type $2$ complex. B: Directed tree rooted at a type $1a$ complex. C: Directed tree rooted on a cycle of type $1a$ complexes. D: Directed tree rooted at a type $0$ complex. E: Directed tree rooted at a type $1b$ complex.}\label{F3}
   \end{figure}
 \subsection{RReLU Neural Networks\label{ReLUCRN}}

\label{SingleLayer}
    \begin{figure}[ht!]
    \centering
    \includegraphics[width=0.8\linewidth]{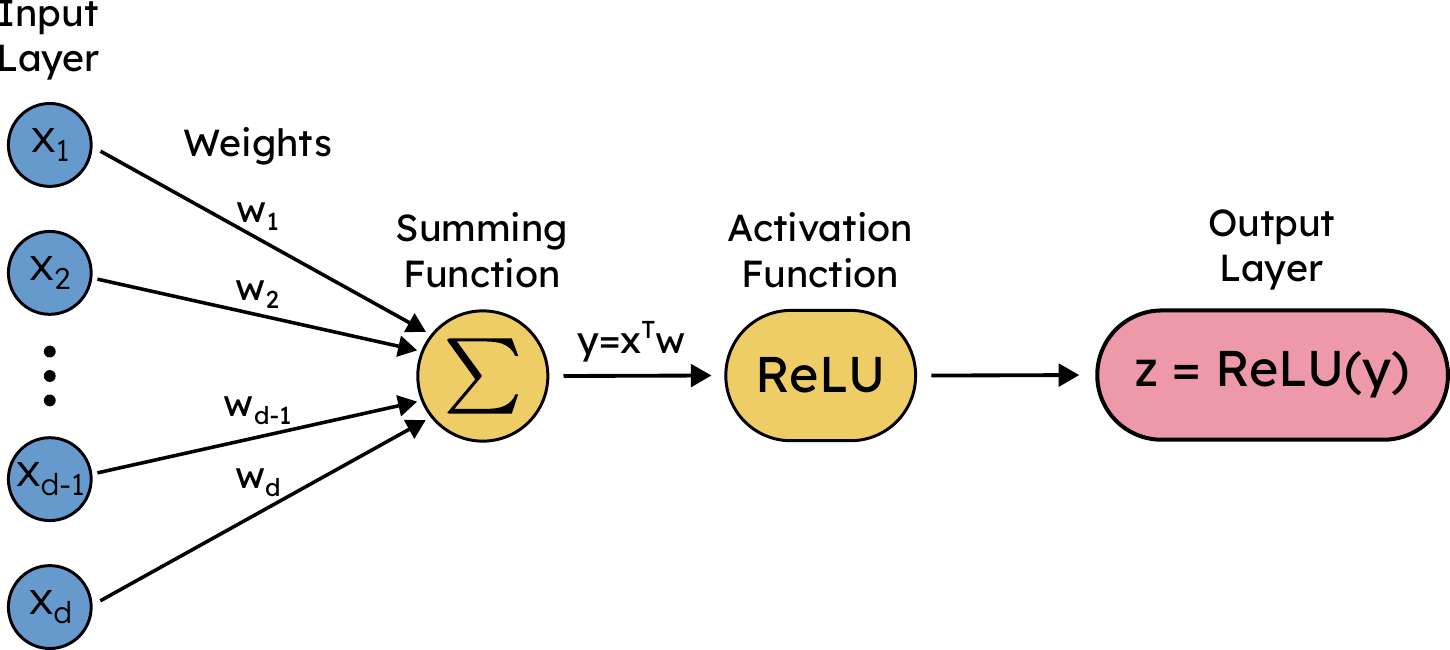}
    \caption{Example of a single layer neural network with ReLU activation function.\label{FigureReLU}}
\end{figure}

The authors of \cite{Vasic} introduced  noncompetitive CRNs that can chemically implement RReLU  neural networks, see Figure \ref{FigureReLU}, with  activation function $f$ given by the ReLU function $f(y)=\max\{0,y\}$. Such NN structures are relevant in molecular networks,  e.g. 
for signaling cascades, see  \cite{Hellingwerf}.
 It has been observed that such linear rectifier neurons exhibit good learning performances
and can reproduce naturally observed sparse data, see, e.g., \cite{bengio}. 
RReLU nets form a special class of ReLU NN that have rational weights with
$w_i = \pm p_i/l_i$, $i=1,\cdots , d$, for integers $p_i$ and $l_i$. 
 To allow the rewriting of NN in terms of noncompetitive CRNs,  the authors of \cite{Vasic} adopt the so-called \emph{dual-rail representation} to encode 
positive and negative parts of
neural net input variables, by introducing CRN species $X^+$ and $X^-$.  When  $X=x \in \mathbb{Z}$, $x= x^+-x^-$ is the difference  of two states of the related CRN
where both $x^+$ and $x^-$ count the number of molecules of species $X^+$ and $X^-$.
The species set is $\mathcal{S}=\{1,\ldots,d,N=d+1\}$
corresponding to the species symbols $X_1^+,\ldots,X_d^+$ and $Y^+$.
Following \cite{Vasic}, RReLU neural networks can be expressed  as a noncompetitive CRN using for all $i \in \{1,...,d\}$, the reactions:
$$
        l_i X_i^+ \longrightarrow p_i Y^+,\ \ 
    l_i X_i^- \longrightarrow p_i Y^-,
$$
when $w_i >0 $,
and
$$
        l_i X_i^+ \longrightarrow p_i Y^-,\ \ 
    l_i X_i^- \longrightarrow p_i Y^+,
$$
when $w_i <0 $. The output of the NN is obtained by applying the ReLU function on $\sum_i w_i x_i^\pm (0)$, for an input $x(0)$. The CRN implementation of this sum is obtained, e.g, by repeated executions of the reactions $l_i X_i^\pm \longrightarrow p_i Y^\pm$ starting with  $x_i^\pm(0)$ units of species $X_i^\pm$. During each reaction, stacks of $l_i$ units of species $X_i^\pm$ are converted into stacks of $p_i$ units of species $Y^\pm$, so that, repeating such reactions as many times as possible, the number of $Y^\pm$ units is increased by
$[x_i^\pm (0)/l_i]p_i$, where $[\cdot]$ is the integer part, leaving $x_i^\pm(0)-[x_i^\pm (0)/l_i]$ units of species $X_i^\pm$. One can thus see each term of the NN sum  $w_i x_i^\pm(0)= p_i/l_i x_i^\pm(0)$ as
$[x_i^\pm (0)/l_i]p_i$. Similar conversion process lead to the CRN implementation 
of the ReLU function $f(y)=\max\{0,y\}$ as the difference $z^+-z^-$ of the abundances of species $Z^\pm$ of the following noncompetitive CRN
\begin{align*}\label{ReLU}
    &Y^+ \longrightarrow M + Z^+, \\ 
    &Y^- + M \longrightarrow Z^-.
\end{align*}

 \subsection{The Abelian sandpile model (ASM) as noncompetitive CRN\label{s.SandA}}

 We show that well-known ASMs can be seen as special examples of noncompetitive CRNs. The study of the algebraical and dynamical properties of the ASM
leads to the notion of (finite) Abelian networks (AN) \cite{bond2016abelian1,bond2016abelian2,bond2016abelian3}. We will recall the main mathematical results of (finite) ANs theory in Section 
\ref{s.AbelianNet}. 

 \subsubsection{The conservative ASM}
 The Abelian sandpile model was originally introduced in physics as a model for avalanches \cite{Bak,Dhar}, and has since been extensively studied in mathematics  and  physics
 \cite{Biggs2,meester2003abelian}. Given an directed connected graph of node set $\mathcal{S}$, 
 we write $j \sim i$ to indicate that the nodes $i$ and $j$ are nearest neighbors.
 We associate a reaction to each node $i\in\mathcal{S}$ of source complex
 $\nu_i = outdeg(i) e_i=\Gamma^r_{.,i}$, where $outdeg(i)$ denotes the out-degree of $i$, 
 and of product complex $\nu'_i = \sum_{j\sim i} e_j=\Gamma^p_{.,i}$. 
In our CRN interpretation of the model, nodes of the graph correspond to species of the CRN, which is noncompetitive with $N=d=p$.
  In this setting, the $U_i$
 can be identified with singletons $\{i\}$.
The source complex $\nu_i$ is of type 1a, while the product complex $\nu'_i$ is of type 2 when $outdeg(i)\ge 2$, and of type 1a otherwise.  When $outdeg(i)\ge 2$ for all $i\in \mathcal{S}$, paths in the graph of complexes reduce to single edges.
Let $\Gamma$ denote the stoichiometric matrix, which in this particular case, satisfies $\Gamma^T = L$, where $L$ is the graph  Laplacian, with ${\rm dim}({\mathcal T})={\rm rank}(L)=N-1$.

\subsubsection{Dissipative ASM\label{s.SandB}}

Most models of ASM use a dissipative version in which, for example, nodes of a $d$-dimensional lattice that are located at its boundary redistribute sand to the outside. This corresponds to considering the reduced Laplacian with open boundary conditions. This implies the existence of sink species in the corresponding CRN.
 
 
 \subsubsection{Toppling networks\label{s.SandC}  }
 
 Toppling networks \cite{Gabrielov,Biggs}   are similar to the ASM with $U_i=\{i\}$, for $i\in \mathcal{S}$, with source complex
 $\nu_i =\Gamma_{ii}^r e_i =\Gamma^r_{.,i}$ and product complex $\nu'_i = \sum_{j\sim i} e_j=\Gamma^p_{.,i}$, for positive stoichiometric coefficients  $\Gamma_{ii}^r \in\N_+$. 
 Given a state $x\in \N^N$, the reaction $\Gamma_{ii}^r e_i\to \nu'_i$ leads to the new state
 $x+\nu'_i-\Gamma_{ii}^r e_i$.

 \subsubsection{Generalized toppling networks \label{s.gSand}. }
We extend toppling networks, while still assuming that $U_i=\{i\}$, for $i\in \mathcal{S}$, and that source complexes satisfy $\nu_i =\Gamma_{ii}^r e_i =\Gamma^r_{.,i}$. The product complexes are now allowed to be more general, and take the form $\nu'_i = \sum_{j\sim i} \Gamma^p_{ji}e_j=\Gamma^p_{.,i}$ where for all $j\ne i$, the non-null coefficients $\Gamma_{ji}^p$ are no longer restricted to be equal to $1$ but must satisfy $\Gamma_{ji}^p \le \Gamma_{jj}^r$.

 
\section{Markov chains associated with CRNs}\label{s.markov_chains}
\subsection{Reminders on Markov chains}
We recall here the notion of recurrence for both discrete-time and continuous-time Markov chains.
\begin{definition}
    Let $(X_t)_{t\in \mathbb{R}}$ be a CTMC (resp. $(X_n)_{n\in \mathbb{N}}$ be a discrete-time Markov chain) on a countable state space $\mathcal{S}$. For $i \in \mathcal{S}$, we define the return probability $m_i$ as follows:
\[
m_i = \mathbb{P}(X_t = i \text{ for some } t > 0 \mid X_0 = i) 
\quad 
\text{(resp. } m_i = \mathbb{P}(X_n = i \text{ for some } n \geq 1 \mid X_0 = i)\text{)}.
\]
    If $m_i=1$, the state $i$ is said to be \emph{recurrent}; if $m_i<1$, the state $i$ is said to be \emph{transient}.

    \end{definition}

This terminology reflects the long-term behavior of the process: If a state is recurrent, then the process is guaranteed to return to it infinitely often, almost surely. In the theory of Markov chains, recurrent states often appear as part of closed communicating classes, which the process cannot escape once entered. This justifies the name \emph{recurrent}, as the system keeps revisiting such states over time.
 \subsection{The usual CTMC}
 In most applications in systems biology and physics, the numbers of molecules of each species in a CRN are random processes whose time evolution is modeled by a Continuous-Time Markov chain (CTMC).
  Let 
 \begin{equation}\label{classic_chain}
     X(t)=(X_1(t),\cdots, X_N(t))\in\N^N,
 \end{equation}
  be the random vector representing the molecular counts of the various species at time $t$. The process $(X(t))_{t\ge 0}$ is a Markov chain with jumps, for reactions 
 $(\nu, \nu')\in\mathcal{E}$ with $l=\nu'-\nu$
 
 \begin{equation*}
     x \longrightarrow  x+l,  
 \end{equation*}
  where $x=(x_1,\cdots,x_N)$ is a state of the CRN.
 We refer the reader to \cite{anderson2015stochastic} for a comprehensive treatment of stochastic reaction networks.
 Such time-continuous stochastic chemical kinetics are widely used, e.g., in molecular biology for small molecular counts systems, and besides giving transition rates provide the times of reactions.
 Static states are absorbing states of the CTMC.
 The authors of \cite{doty2024computational} obtained results on absorption time for CTMC versions  of rate independent CRNs using mass-action transition rates.


\subsection{Sandpile Markov chain for noncompetitive CRNs\label{s.SandMC-CRN}}


The authors of \cite{chen2023rate} and \cite{Vasic} proposed to study the structural properties of CRNs using a simplified nondeterministic dynamical model
which was designed to isolate the effect of stoichiometry from the effect of transition rates such as the mass-action rates used in the usual CTMC. The model was used to explore the set of reachable states from static states 
$q\in \operatorname{Stat}\cH$
using only the CRN's stoichiometric structure.
Remark \ref{r.unique-static} and Theorem \ref{ThVasic} show that, for $i\in\mathcal{S}$,
the  stoichiometric compatibility class $(q+e_i +\mathcal{T})\cap\Z_{\ge 0}^N$ of  entirely execution bounded noncompetitive CRN
contains a unique element $q'$ of $\operatorname{Stat}\cH$, which is  uniquely defined given $q+e_i$.

\begin{definition}\label{Rmap}
Consider an entirely execution bounded noncompetitive CRN.
Given $i\in\mathcal{S}$, let $R_i:\ \operatorname{Stat}\cH \to\operatorname{Stat}\cH $ be the map 
defined by
$\{R_i(q)\} = (q+e_i +\mathcal{T})\cap\Z_{\ge 0}^N \cap \operatorname{Stat}\cH$
 and
$q+e_i \hookrightarrow R_i(q)$. 
\end{definition}

 As stated in the Introduction, CRNs of interest in biochemistry and systems biology are embedded in complex networks where the outputs of submodules feed other  CRNs. Given $q\in \operatorname{Stat}\cH$, we model the response of the CRN to such perturbations by considering the effect of adding a molecule from a randomly chosen species to static state $q$,
      taking inspiration from the literature on the Abelian sandpile model \cite{klivans2018mathematics,Biggs2,bond2016abelian3,meester2003abelian}; see also the examples given in (\ref{s.SandA}, \ref{s.SandB},\ref{s.SandC}). \\
\begin{definition}\label{d.CRNsandpile}[The CRN sandpile Markov chain]
  The state space of the CRN sandpile Markov chain is the set of static states ${\rm Stat}\ \cH$, and its transition  mechanism is given by
$q_{n+1}=R_{i_{n+1}}(q_n)$ where $i_1,\ i_2\ldots$ are independent of law $\mu$, where
$\mu$ is a probability measure on $\mathcal{S}$.  
\end{definition}

We will give more information on this Markov chain in Section 
\ref{s.sandpile-generalized} when focusing on the class of generalized toppling networks defined in  Section \ref{s.gSand}.
The resulting Markov chains have recurrent states that will appear infinitely often along the orbits of the dynamics, while the other static states will not persist in the long run. Besides the plausible relevance of recurrent static states in systems biology, one can also see their importance in looped and recurrent neural networks  since, following \cite{Vasic}, one can translate ReLU neural networks as noncompetitive CRNs, see Section \ref{SingleLayer}.


    \section{Abelian networks (AN)\label{s.AbelianNet}}
The purpose of this section is to review the fundamental definitions and results of Abelian network theory \cite{bond2016abelian1,bond2016abelian2, bond2016abelian3} that will serve as tools for the remainder of the paper. We then show how noncompetitive execution bounded CRNs can be expressed within this framework.
\subsection{Basic mathematical notions from AN theory\label{s.BasicAN}}
Let $G=(V,E)$ be a directed graph. A network on $G$ is a set of processors $(\mathcal{P}_v)_{v\in V}$ associated with each vertex $v$. Each processor $\mathcal{P}_v$ contains an input alphabet $\mathcal{A}_v$ and a state space $\mathcal{Q}_v$. \\

The time evolution  of the processor states is governed by transition functions
 $T_v: \mathcal{A}_v \times \mathcal{Q}_v \rightarrow \mathcal{Q}_v$  and  message-passing functions $T_{(v,v')}: \mathcal{A}_v \times \mathcal{Q}_v \rightarrow \mathcal{A}^*_{v'}$, for each $(v,v')\in E$, where $\mathcal{A}^*_{v'}$ denotes the free monoid of all finite words over the alphabet $\mathcal{A}_{v'}$. 
 Suppose that processor $\mathcal{P}_v$ receives an input $a\in \mathcal{A}_v$ while in state $q\in \mathcal{Q}_v $. Then it updates its internal state $q$
 to the new state $T_v(a,q)\in \mathcal{Q}_v  $ and, for each edge $(v,v')\in E$, sends $T_{(v,v')}(a,q)\in\mathcal{A}_{v'}  $ to its neighbouring processors $\mathcal{P}_{v'}$ .

A processor $\mathcal{P}_v$ is called abelian if permuting the input letters does not change the resulting state of the processor, nor its outputs, up to permutation. 
For a word $w \in \mathcal{A}^*_v$, we denote by $|w| \in  \mathbb{N}^{\mathcal{A}_v}$ the vector counting the number of occurrences of each letter of $\mathcal{A}_v$ in $w$. Then the processor $\mathcal{P}_v$ is abelian if for any words $w,w' \in \mathcal{A}_v^*$ such that $\lvert w \rvert= \lvert w' \rvert$, for all $q \in \mathcal{Q}_v$ and any $v\to v' \in E$, we have:
\begin{equation}\label{e.abelianness}
    T_v(w,q) = T_v(w',q) \quad \mbox{ and } \quad \lvert T_{(v,v')}(w,q)  \rvert= \lvert T_{(v,v')}(w',q) \rvert
\end{equation}.
\begin{definition}\cite{bond2016abelian1}
An Abelian network (AN) $\mathcal{N}$ on a graph $G=(V,E)$ is a collection of abelian processors  $(\mathcal{P}_v)_{v \in V}$.
The total alphabet and state space are respectively $\mathcal{A} = \cup_{v \in V} \mathcal{A}_v$ and $\mathcal{Q} = \prod_{v \in V} \mathcal{Q}_v$ .
The \textit{state} of an Abelian network $\mathcal{N}$ is denoted $x.q$, where $q \in \mathcal{Q}$ describes the internal state and $x \in \mathbb{N}^{\mathcal{A}}$ is a vector counting the number of letters waiting to be processed.
\end{definition}
For $v \in V$ and $a \in \mathcal{A}_v$, consider the map $t^v_a: \mathcal{Q} \to \mathcal{Q}$:
$$
t^v_a(q)_u = \left\{
    \begin{array}{ll}
        T_v(a,q_v) & \mbox{if } u=v, \\
        q_u & \mbox{otherwise.}
    \end{array}
\right.
$$
Following \cite{bond2016abelian2}, one can then define mappings $\pi_a$, $a\in\mathcal{A}$, that act on states in three steps : it modify the internal state $q$ to $t^v_a q$
where $v$ is the unique node such that $a\in \mathcal{A}_v$, decrement $x_a$ by one and increment each $x_b$ by the number of $b$ letters passed by 
$T_{(v,v')}$ with $b\in\mathcal{A}_{v'}$.
\begin{definition}\label{ExecutionAbelian}
AN executions are words $w=a_1\ldots a_r\in \mathcal{A}^{*}$ detailing the order in which the letters are processed. 
Let $\pi_w(x.q)$ be the result of executing $w$ starting from $x.q$, with
$\pi_w = \pi_{a_1}\circ \ldots \circ \pi_{a_r}$.
The AN execution is legal if, at each step, the letter to be processed is available, that is, its count is strictly positive at that step. The AN execution is complete if it removes all letters from the network.
\end{definition}

\begin{example}[ASM and toppling  networks]\label{toppling}
Consider  the ASM of Section \ref{s.SandA}; see \cite{bond2016abelian1} for more details. Given some graph $G=(V,E)$, the ASM corresponds to an Abelian network where, for each $i\in V$, $\mathcal{A}_i = \{i\}$,  $\mathcal{Q}_i =\{0,1,\ldots,outdeg(i)-1\}$, and  where the  transition functions are
$$T_i(a,q) = (q+1)\ {\rm mod}\ outdeg(i),$$
and, for each edge $(i,j)\in E$,
$$
T_{(i,j)}(a,q)= \left\{
    \begin{array}{ll}
        \varepsilon & \mbox{if } q< outdeg(i)-1, \\
         j & \mbox{if } q= outdeg(i)-1,
    \end{array}
\right.$$
where $\varepsilon \in \mathcal{A}^*$ denotes the empty word, indicating that no message is sent.  Toppling networks \cite{bond2016abelian1} are generalization of such Abelian networks, which have the same transitions and message passing functions but with thresholds $r_v$
at vertex $v$ that can be different of its outdegree. Notice that the model proposed in Section \ref{s.gSand} is also a generalization of toppling networks.
\end{example}
The remaining part of this section provides an overview of the relevant framework from Abelian networks \cite{bond2016abelian1,bond2016abelian2,bond2016abelian3}. We will later reinterpret these results in the context of noncompetitive CRNs.\\

An Abelian network {\bf halts on all inputs} if, for all $x\in \mathbb{N}^{\mathcal{A}}$ and $q \in \mathcal{Q}$, the state $x.q$ admits a complete legal AN execution. In this case the {\bf odometer} is defined as 
$$[x.q] = |w| ,$$
where $w$ is any complete legal AN execution.\\

Given $x\in \mathbb{N}^{\mathcal{A}}$,  let $\pi_x$ denote $\pi_w$ for a word $w$ such that $\vert w\vert = x$ with $\pi_x(x.q)= x'. q'$; set (See (\ref{local_action}) from the Appendix)
\begin{equation}
   x \triangleright q = \pi_x (x.q)\in \mathbb{N}^{\mathcal{A}}\ {\rm x}\ \mathcal{Q}.
\end{equation}

This is called the \textit{local action}, since each processor processes only letters added at its own vertex, not those passed from other processors. 
Starting from state $q\in\mathcal{Q}$, $x \triangleright q $ is obtained by adding $x_i$ letters $i\in \mathcal{A}$ and by processing each letter once.
This distinguishes it from the \textit{global action} defined as: 
\begin{equation}\label{global_action}x \triangleright \triangleright q = q'\in \mathcal{Q}\end{equation}
which denotes the final (processor) state  obtained by processing the word $x$,  for a complete and legal AN execution $w$ such that $\pi_w(x.q)= 0. q'$.
The final state $q_v^r$ of each processor is obtained as
$$q_v^r = \Big(\prod_{a\in\mathcal{A}_v}(t_a^v)^{[x.q]_a}\Big) q_v.$$
For $a \in \mathcal{A}$, consider the map
 $\tau_a: \mathcal{Q} \to \mathcal{Q}$:
 \begin{equation}\label{BasicMap}
\tau_a (q) = 1_a \triangleright \triangleright q. 
\end{equation}
The set $M_v = \left\langle t_a^v \right\rangle_{a \in \mathcal{A}_v}$ is a submonoid of $\text{End}(\mathcal{Q}_v)$ called the transition monoid of processor $ \mathcal{P}_v$ and the set $M=\left\langle\tau_a\right\rangle_{a \in \mathcal{A}} \subset \text{End}(\mathcal{Q})$ is a global monoid. Let $t_v:\ \N^{\mathcal{A}_v}\to M_v$ be the monoid homorphism that sends the basis element $1_a$ to the  $t_a^v$. One can show (see \cite{bond2016abelian2}) that this extends to an action 
$\Z^{\mathcal{A}_v}\ {\rm x}\ e_v Q_v \to e_v Q_v$, where $e_v$ is the minimal indempotent of the local monoid $M_v$. Let $K_v$ be the kernel of this (group) action, and set
\begin{equation}\label{e.kernel}
    K = \prod_{v \in V} K_v \subset \mathbb{Z}^{\mathcal{A}}
\end{equation}
which is the {\bf total kernel} (see Definition \ref{TotalKernel}). 

Let $e$ be the mimimal indempotent of the global monoid $M$. A state 
$q\in\mathcal{Q}$ is {\bf recurrent} when $q=eq$. The set of recurrent states is denoted by $\text{Rec }\mathcal{N}  $, see Definition \ref{d.alg_recurrent}. 
Such states appear as recurrent states of  Markov chains associated with AN, see Section \ref{s.SandMC-AN}.

The {\bf production matrix}  $P$ (Definition \ref{d.prodmatrix})  has entries $P_{ab}$ representing the average number of letters $a$ produced when processing the letter $b$. $L$ denotes the {\bf reduced Laplacian} , and 
$D$ the diagonal matrix whose entries are the firing thresholds of the vertices (see Definition \ref{reduced Laplacian}).
More details are given in the Appendix, and full algebraical details can be found in \cite{bond2016abelian1,bond2016abelian2,bond2016abelian3}.\\
\vspace{1mm}

To find which  states are recurrent,   
Levine \textit{et al.} \cite{bond2016abelian3} generalized Dhar's burning algorithm to any finite irreducible abelian network $\mathcal{N}$ (see Definition \ref{d.irreducibility} and \ref{d.irreducibility}). The algorithm looks for so-called \emph{burning elements} (see Definition \ref{d.burning_element}), using Lemma \ref{l.procedure}. Such an element returns the network to its initial state if and only if the state is recurrent (Definition \ref{d.alg_recurrent}). This gives  a simple test to determine if a state is recurrent (Theorem \ref{t.burning_algo}). Here is a short description of the burning algorithm for ANs, see Appendix for more details. \\

\begin{algorithm}[H]
\caption{Burning Algorithm for ANs  }
\label{alg:burning}
\begin{algorithmic}[1]\label{burning_algorithm}
\State \textbf{Input:} a state  $q\in\mathcal{Q}$.
\textbf{Output:} Whether $q$ is  recurrent.
\State \textbf{Step 1: Find a burning element $y$.}
\State Initialize $y \gets 1_{\{1,...,N\}}$  
\While{$Ly \not\ge 0_N$}
    \State Choose an index $a \in A$ such that $(Ly)_a < 0$
    \State Increase $y_a \gets y_a + 1$
\EndWhile
\end{algorithmic}

\vspace{3mm}

\begin{algorithmic}[1]
\State \textbf{Step 2: Test if  $q\in\mathcal{Q}$} is recurrent.
\State Compute $k \gets Dy$
\If{$(I - P)k \triangleright\triangleright q = q$}
    \State \textbf{return} ``$q$ is recurrent``
\Else 
    \State \textbf{return} ``$q$ is not  recurrent''
\EndIf
\end{algorithmic}

\end{algorithm}
We follow this algorithm in the examples of Section \ref{s.examples}.

\subsection{Sandpile Markov chain for Abelian networks\label{s.SandMC-AN}}

We follow \cite{bond2016abelian3,meester2003abelian} by considering AN sandpile Markov chain 
\begin{definition}\label{d.ANsandpile}[AN sandpile Markov chain]
  The AN sandpile Markov chain   $(q_n)_{n\geq 1}$, $q_n \in \mathcal{Q}$, 
  is defined by the transition mechanism
  \begin{equation}\label{next_state_markov}
        q_{n+1} = 1_{a_{n+1}} \triangleright \triangleright q_n = \tau_{a_{n+1}}(q_n),
    \end{equation}
where the sequence $(a_n)_{n\ge 1}$ is i.i.d. from a distribution $\mu$ on $\mathcal{A}$.
\end{definition}
Intuitively, a letter $a_{n+1}$ is chosen at random in $\mathcal{A}$ and the unique processor $\mathcal{P}_v$ with $a_{n+1} \in \mathcal{A}_v$ processes this letter to arrive at the new state
$1_{a_{n+1}} \triangleright \triangleright q_n$.
The theory of Abelian sandpile \cite{meester2003abelian} and Abelian networks \cite{bond2016abelian1,bond2016abelian2,bond2016abelian3}, shows that, under some conditions on $\mu$, the stationary distribution of this Markov chain is uniform  on the set $\text{Rec }\mathcal{N}$, see   Lemma \ref{l.algebraically_recurrent}
of the Appendix.
 $\text{Rec }\mathcal{N}  $ is in a one to one relation with the {\bf  critical group} $\text{Crit }\mathcal{N} $ (see Definition \ref{CriticalGroup}) of the AN which is such that
 \begin{equation}\label{e.crit}
    \text{Crit }\mathcal{N} \simeq \mathbb{Z}^\mathcal{A} / (I-P)K,
\end{equation}
see Theorem \ref{critical_group_2}.
This permits to give a formula for the number of recurrent states that involves the determinant of the reduced Laplacian matrix (Theorem \ref{t.number_recurrent_states}).
One  can also get information on the typical number of reactions involved in a single step of the Markov chain:
Let $q$ be a uniform random element of $\text{Rec }\mathcal{N} $, Theorem 3.7 of  \cite{bond2016abelian3} shows that, for all $x \in \mathbb{N}^{\mathcal{A}}$, the mean value of the odometer is given by
$$\mathbb{E}[x.q] = (I - P)^{-1}x,$$
where $\mathbb{E}$ denotes the mathematical expectation with respect to the uniform measure on $\text{Rec }\mathcal{N} $.
More details are provided in the Appendix.

\section{Noncompetitive CRNs and AN\label{s.NonCompetAbelianNet}}
We assume that $U_0=\emptyset$  and show that  execution bounded noncompetitive CRNs can be associated with Abelian Networks that halt on all inputs and use this fact to derive interesting and useful properties.
When $U_0=\emptyset$, each source complex $\nu_k$ is associated in a one  to one correspondence to some $U_k$, $k=1,\ldots, d=p$, with
${\rm supp}(\nu_k)=U_k$, see Section \ref{s.structural}, where $\nu_k$ is  the $k^{th}$ column of the reactant matrix $\Gamma^r$.
We associate a processor $\mathcal{P}_k$ to each source complex $\nu_k$. Recall that complexes $\nu_k\in \mathcal{C}$ are written as non-negative integer linear combinations of species, represented as vectors in $\Z_{\geq 0}^N$, with
$$
\nu_k =  \sum_{i=1}^N \Gamma^r_{ik} e_i.
$$
By noncompetitivity, this source complex is the reactant in a unique reaction $ \nu_k \to \nu'_k$, where:
$$
\nu'_k =  \sum_{i=1}^N \Gamma^p_{ik} e_i.
$$
The alphabet of processor $\mathcal{P}_k$  is $\mathcal{A}_k = U_k $ and the different alphabets are finite and disjoint by construction.

The corresponding AN state space   $\mathcal{Q}_k \subset \mathbb{N}^{N}$ of processor $\mathcal{P}_k$ is the following subset of the set of static states
\begin{equation}\label{StateSpace}
\mathcal{Q}_k = \{\sum_{j\in U_k} \gamma_{jk} e_j,\ \gamma_{jk}\in\N,\ \exists i\in U_k \hbox{ with }\gamma_{ik}< \Gamma^r_{ik}\},
\end{equation}
and the total state space is given by 
$$ \mathcal{Q}=\mathop{\oplus}\limits_{k=1\ldots p} \mathcal{Q}_k = {\rm Stat}\ \cH .$$ 
The following transition function $T_{U_k}: \mathcal{A}_{k} \times \mathcal{Q}_{k} \rightarrow \mathcal{Q}_{k}$ is associated to each processor $\mathcal{P}_k$. When being in state $q^k\in \mathcal{Q}_k$ and processing $j\in U_k$, the transition function of processor $\mathcal{P}_k$ is given by:
 \begin{equation}\label{Transition}
 T_{U_k}(j,q^k) = \left\{
    \begin{array}{ll}
           
    q^k +e_j -\nu_k & \text{if}  \quad q^k_j +1 = \Gamma^r_{jk} \text{ and } \forall i\in U_k \setminus \{j\}:  q^k_i \ge \Gamma^r_{ik}, \\
    q^k +e_j & \mbox{ otherwise} ,
    \end{array}
\right.
\end{equation}
We also define, for each pair $(U_l,U_k)$ with ${\rm supp}(\nu'_l)\cap \mathcal{A}_k \neq \emptyset$, a  message-passing function $T_{(U_l,U_k)}: \mathcal{A}_{l} \times \mathcal{Q}_{l} \rightarrow \mathcal{A}_{k}^*$. When being in state $q^l \in \mathcal{Q}_l$ and processing $j\in U_l$, the message sent from processor $\mathcal{P}_l$ to processor $\mathcal{P}_k$ is given by:

\begin{equation}\label{MessagePassing1}
T_{(U_l,U_k)}(j,q^l) = \left\{
    \begin{array}{ll}
     \prod_{i\in U_k} i^{\Gamma_{il}^p}  & \text{if}  \quad q^l_j +1 = \Gamma^r_{jl} \text{ and } \forall i\in
    
    U_k \setminus \{j\}: \ q^l_i \ge \Gamma^r_{il},\\
          \epsilon  & \mbox{otherwise},
          \end{array}
\right. 
\end{equation}
where, for $i\in U_k$, $i^{\Gamma_{il}^p}$ is the word of $\mathcal{A}_{k}^*$
composed of $\Gamma_{il}^p$ letters $i$
and $\epsilon$ denotes the empty word.
The above construction associates a natural AN to an arbitrary noncompetitive CRN.
The general theory  developed in \cite{bond2016abelian1,bond2016abelian2,bond2016abelian3} for finite AN
 (see Sections \ref{s.BasicAN} and \ref{BasicsAbelianNetwork} for an overview)
 cannot be applied as such  since the state spaces $\mathcal{Q}_k$ are infinite for noncompetitive CRNs involving multi-species source complexes, see e.g. the multi-species   implementation  of the ReLU function of Section
\ref{ReLUCRN}.  Section \ref{s.single_species} focus on the class of generalized toppling networks which was considered in Example \ref{s.gSand}, and show that the related CRN and AN sandpile Markov chains defined in Sections \ref{s.SandMC-CRN} and \ref{s.SandMC-AN} are identical.

\subsection{Generalized toppling networks}\label{s.single_species}
We consider the CRNs defined in Example \ref{s.gSand}, which are generalizations of toppling networks, see e.g. \cite{bond2016abelian1}.
We assume again that $U_0=\emptyset$ and that the CRN is entirely execution bounded. In this setting, $N=d=p$.
We associate a processor $\mathcal{P}_k$ to each source complexes $\nu_k \in \mathcal{C}_S$, which are of the form $\nu_k= \Gamma_{kk}^re_k$. By noncompetitivity, this source complex is the reactant in a unique reaction $ \nu_k \to \nu'_k$, where the product complexes $\nu'_k\in \mathcal{C}$ are written as non-negative integer linear combinations of species, represented as vectors in $\Z_{\geq 0}^N$, with
$$
\nu'_k =  \sum_{\substack{i=1 \\ i \ne k}}^{N} \Gamma^p_{ik} e_i.
$$
We recall that we imposed $ \Gamma_{ik}^p \leq \Gamma_{ii}^r $ for all $ i,k \in \{1,...N\}$.\\
The alphabet of processor $\mathcal{P}_k$ is $\mathcal{A}_k = \{k\} $ and the different alphabets are finite and disjoint by construction.
The corresponding state space   $\mathcal{Q}_k \subset \mathbb{N}^{N}$ of processor $\mathcal{P}_k$ is given by

\begin{equation}\label{StateSpaceNew}
\mathcal{Q}_k =\{\gamma_{kk}e_k,\ \gamma_{kk}\in\N,\ \gamma_{kk} < \Gamma_{kk}^r\},
\end{equation}

and the total state space is given by $ \mathcal{Q}=\mathop{\oplus}\limits_{k=1,\ldots, d} \mathcal{Q}_k= {\rm Stat}\ \cH$, which is finite. 
 The following transition function $T_{U_k}: \mathcal{A}_{k} \times \mathcal{Q}_{k} \rightarrow \mathcal{Q}_{k}$ is associated to each processor $\mathcal{P}_k$. When being in state $q^k\in \mathcal{Q}_k$ and processing $k$, the transition function of processor $\mathcal{P}_k$ is given by:
 \begin{equation}\label{TransitionNew}
 T_{k}(q^k) = \left\{
    \begin{array}{ll}
           q^k +e_k- \Gamma_{kk}^r e_k & \text{if}  \quad q^k_k = \Gamma_{kk}^r -1,  \\
    q^k +e_k & \mbox{ otherwise} .
    \end{array}
\right.
\end{equation}
We also define, for each pair $(l,k)$ with ${\rm supp}(\nu'_l)\cap \mathcal{A}_k \neq \emptyset$, and $k \neq l$ a  message-passing function $T_{(l,k)}: \mathcal{A}_{l} \times \mathcal{Q}_{l} \rightarrow \mathcal{A}_{k}^*$. When being in state $q^l \in \mathcal{Q}_l$ and processing species $l$, the message sent from processor $\mathcal{P}_l$ to processor $\mathcal{P}_k$ is given by:

\begin{equation}\label{MessagePassing}
T_{(l,k)}(q^l) = \left\{
    \begin{array}{ll}
    k^{\Gamma^p_{kl}}  & \text{if}  \quad q^l_l =\Gamma_{ll}^r -1,\\
          \epsilon  & \mbox{otherwise}.
          \end{array}
\right. 
\end{equation}

\subsection{Sandpile Markov chain for generalized toppling networks}\label{s.sandpile-generalized}

Let $(q_n)_{n\ge 0}$, $q_n\in \operatorname{Stat}\cH$, be the CRN sandpile Markov chain which has been considered in Section \ref{s.SandMC-CRN}, Definition \ref{d.CRNsandpile}. Given that $q_n = q$, 
$q_{n+1}$
is obtained by
choosing a species $i_{n+1}\in\mathcal{S}$ according to a probability measure $\mu$ on $\mathcal{S}$, and then letting the CRN evolve starting from $q+e_{i_{n+1}}$ until a static state $q_{n+1}=q'\in \operatorname{Stat}\cH$, 
with $q_{n+1}=R_{i_{n+1}}(q_n)$
is obtained, see Definition \ref{Rmap}.

We show below that, for generalized toppling networks, the maps $R_i$ 
correspond to the maps $\tau_i$ (see (\ref{BasicMap})) that characterize the transition mechanism
 $$   q_{n+1} = 1_{i_{n+1}} \triangleright \triangleright q_n =\tau_{i_{n+1}}(q_n),$$
of the  chain defined in Section \ref{s.SandMC-AN},
see (\ref{next_state_markov}). We  prove in Lemma \ref{l.map} 
 that $R_i(q)\equiv 1_{i} \triangleright \triangleright q =\tau_i(q)$, $\forall i\in\mathcal{S}$ and $\forall q\in \operatorname{Stat}\cH = \mathcal{Q}$. Both CRN and AN sandpile Markov chains have therefore the same transition mechanism and live on the same state space, and are  thus identical. \\

Consider a finite CRN execution ${\mathcal F}_{qq'}=(c^0,\ldots,c^m)$ where $c^0 = q+e_i$ and $c^m=q'\in \operatorname{Stat}\cH$. If $m>1$, then
 $c^l \not\in\operatorname{Stat}\cH$, $1\le l<m$.  Each step $c^k \to c^{k+1}$ of ${\mathcal F}_{qq'}$
corresponds to a reaction $\nu_{j_k}'-\nu_{j_k}$ for some $j_k\in \mathcal{S}$, so that
$$c^l = c^0 +\sum_{k=1}^ l(\nu_{j_k}'-\nu_{j_k}).$$
We focus on CRNs that are entirely execution bounded generalized toppling networks by linking CRN  and AN executions.

\begin{lemma}\label{l.map}
Consider an entirely execution bounded generalized toppling network. Let $q\in \operatorname{Stat}\cH$. Let 
$c=q+e_i$  for some $i\in\mathcal{S}$. Let $q'\in \operatorname{Stat}\cH$ be the unique static state given by
$q'=R_i(q)$ with $q+e_i \hookrightarrow q'$. 
Consider a finite CRN execution ${\mathcal F}_{qq'}=(c^0,\ldots,c^m)$ where $c^0 = q+e_i$ and $c^m=q'\in \operatorname{Stat}\cH$, with
$$c^l = c^0 +\sum_{k=1}^ l(\nu_{j_k}'-\nu_{j_k}),\ l=1,\ldots,m,$$
and $c^l \not\in \operatorname{Stat}\cH$, $l<m$.
There is a sequence of legal AN executions $w_l$, $l=1,\ldots, m$, with
$\pi_{w_l}(1_i . q)= x^l.q^l$ such that $x^l+q^l=c^l$,
$l=1,\ldots,m$.
Moreover,
$$1_i \triangleright \triangleright q =R_i(q).$$
\end{lemma}

Recall that
$1_{i} \triangleright \triangleright q$ is the processor state of $\pi_w(1_{i}.q)$ for a complete and legal AN execution $w$ where, for given $i\in\mathcal{S}$, the map $\pi_i$ 
acts by modifying the internal state $q$
to $t_i^i(q)$ such that $t_i^i(q)_i=T_i(q_i e_i)$ and $t_i^i(q)_j = q_j$, $j\ne i$,
decrementing $(1_{i})_i$ by one and incrementing each $(1_{i})_j$, $j\ne i$, by the number of $j$ letters passed by $T_{(i,j)}$.

\begin{proof}

{\bf I. Assume that $q+e_i \in \operatorname{Stat}\cH$}.  In this situation, $T_i(q_i e_i)=q_i e_i +e_i =(q_i+1)e_i$,
so that $t_i^i(q)=q'$.
Using the fact that no letters are passed to other processors, the new AN state is
$0. q'$ and it follows that $\pi_w(1_{i}. q)=0. q'$, so that
$$q' = 1_{i} \triangleright \triangleright q,$$
 for the complete and legal AN execution
 $w=i$.\\

{\bf II. Assume that $q+e_i \not\in \operatorname{Stat}\cH$, and $m=1$}. In this case $c^1=q'\in \operatorname{Stat}\cH$. 
The reaction takes $q+e_i$ to the static state $c^1=q+e_i-\nu_i +\nu_i'$, with
$c^1_i =0$ and $c^1_j = q_j+\Gamma_{ji}^p$, $j\ne i$.
One needs to show that
$0.c^1=\pi_w(1_{i} . q)$ for a complete and legal AN execution $w$. First notice that
$$\pi_i(1_{i}. q)=(\sum_{j\ne i}\Gamma_{ji}^p e_j) . t_i^i(q) = \nu_i'.t_i^i(q),$$
with  $t_i^i(q)_i =0$ and $t_i^i(q)_j= q_j$, $j\ne i$, so that
$t_i^i(q)=q-q_i e_i$.
Using the fact that
$q_j +\Gamma_{ji}^p < \Gamma_{jj} ^r$ since $q'\in \operatorname{Stat}\cH$, we arrive at
$$\pi_k^{\Gamma_{ki}^p}\Big(\sum_{j\ne i}(\Gamma_{ji}^p e_j) .( q-q_i e_i)\Big)=
(\sum_{j\ne k, i}\Gamma_{ji}^p e_j).(q-q_i e_i +\Gamma_{ki}^p e_k),\ k\ne i,$$
and thus
$$\prod_{k\ne i}\pi_k^{\Gamma_{ki}^p}\Big((\sum_{j\ne i}\Gamma_{ji}^p e_j) . (q-q_i e_i)\Big)=
0 . (q-q_i e_i +\sum_{k\ne i}\Gamma_{ki}^p e_k)= 0. q',
$$
and it follows that $\pi_w(1_{i}. q)= 0.c^1 = 0. q'$, where $w=\prod_{k\ne i}k^{\Gamma_{ki}^p}i$, and therefore $1_i \triangleright \triangleright q =q'$.\\

{\bf III. Assume that $q+e_i \not\in \operatorname{Stat}\cH$, and $m>1$}.
 In this case,
$$\pi_i(1_i . q)=\nu_i' . (q+e_i-\nu_i),$$
is an AN state of the form $x.s$ with $x+s=q+e_i+\nu_i'-\nu_i=c^1$. The sum of the two components of the AN state gives thus the CRN state $c^1$  after reaction $\nu_i \to \nu_i'$ starting from $q+e_i$, with $c^1_i =0$.
$m>1$ so that the next reaction occurs at $j_2\ne i$, with
 $q_{j_2}+\Gamma_{j_2i}^p\ge \Gamma_{j_2j_2} ^r$. Applying $\Gamma_{j_2i}^p$ times the map $\pi_{j_2}$  
 to  $\nu_i' . (q+e_i-\nu_i) $  will 
1) provoke the action of $t_{j_2}^{j_2}$, 2) increment the internal state $q+e_i-\nu_i$ by $\Gamma_{j_2i}^p e_{j_2}$, and 3) decrement $\nu_i'$ by $\Gamma_{j_2i}^p e_{j_2}$,
so that the new AN state will by
$$\pi_{j_2}^{\Gamma_{j_2i}^p} (\nu_i' . (q+e_i-\nu_i))=(\nu_i'-\Gamma_{j_2i}^p e_{j_2}+\nu_{j_2}').(q+e_i-\nu_i +\Gamma_{j_2i}^p e_{j_2}-\nu_{j_2}).$$
The sum of its two AN components is $q+e_i+\nu_i'-\nu_i +\nu_{j_2}'-\nu_{j_2}=c^2$ which is the state of the CRN execution  after two steps. In 1), the fact that there is only one  application of $t_{j_2}^{j_2}$ is a consequence of the constraint $\Gamma_{j_2i}^p \le\Gamma_{j_2 j_2}^r$.\\

We have seen that the first statement holds  true when $l=1$ and $l=2$. 
Assume  that
there is a sequence of legal AN executions $w_l$, $l=1,\ldots, k$, with
$\pi_{w_l}(1_i . q)= x^l.q^l$ such that $x^l+q^l=c^l$. We show that the same property holds for $k+1\le m$. The $(k+1)^{th}$ reaction leading to $c_{k+1}$ occurs at some $j_{k+1}\in \mathcal{S} $. Hence, using the relation
$q^k + x^k =c^k$, we have that $x^k_{j_{k+1}}+q^k_{j_{k+1}} \ge \Gamma_{j_{k+1}j_{k+1}}^r$, with
$\Gamma_{j_{k+1}j_{k+1}}^r -q_{j_{k+1}}^k \ge 0$. Then
$$\pi_{j_{k+1}}^{\Gamma_{j_{k+1}j_{k+1}}^r -q_{j_{k+1}}^k}
(x^k . q^k)= x^{k+1}. q^{k+1},$$
where
$x^{k+1}=x^k -(\Gamma_{j_{k+1}j_{k+1}}^r -q_{j_{k+1}})e_{j_{k+1}}+\nu_{j_{k+1}}'$
and $q^{k+1}=q^k +(\Gamma_{j_{k+1}j_{k+1}}^r -q_{j_{k+1}})e_{j_{k+1}} -\nu_{j_{k+1}}$
are such that $x^{k+1}+q^{k+1}=c^{k+1}$ as required.\\

The last statement is obtained as follows: we have seen that there  is a legal execution $w_m$
such that $\pi_{w_m}(1_i . q)= x^m . q^m$ with
$x^m + q^m = c^m = q'\in\operatorname{Stat}\cH$. Hence
$x_i^m + q_i^m < \Gamma_{ii}^r,\ \forall i\in\mathcal{S}$, so that
$$\pi_i^{x_i^m}(x^m . q^m)=(x^m - x_i e_i).(q^m + x_i e_i),\ \forall i\in\mathcal{S}.$$
Setting $w_0 = \prod_{i\in\mathcal{S}} i^{x_i^m} $ and $w_f = w_0 w_m$, we finally arrive at $\pi_{w_f}(1_i .q)= 0. c^m$ for a complete and legal execution $w_f$, proving the last statement.
\end{proof}

\begin{proposition}
    The collection of processors $\mathcal{N}=(\mathcal{P}_k)$ where ${\nu_k \in \mathcal{C}_S}$, 
    of an entirely execution bounded generalized toppling network
    on the directed graph ${\mathcal G}=({\mathcal C},{\mathcal E})$ forms a finite irreducible Abelian network that halts on all inputs.
\end{proposition}
\begin{proof}
    Since for any $\nu_k \in \mathcal{C}_S$, the alphabet $\mathcal{A}_k $ is constituted of a single species $k$, any two words $w, w'$ in $\mathcal{A}_k^* $ such that $\vert w\vert = \vert w' \vert $ are equal i.e. $w = w'$. The two conditions of (\ref{e.abelianness}) are trivially satisfied and thus each processor is abelian.
    To each processor $\mathcal{P}_k $ is associated a monoid action $M_k \times \mathcal{Q}_k \rightarrow \mathcal{Q}_k$, called local action (\ref{local_action}), which is irreducible because for any $q_k, q'_k \in \mathcal{Q}_k$ there exists $n \in \{0,..,\Gamma^r_{kk}\}$ such that $n k \triangleright q_k = q'_k$. We conclude that each processor is irreducible and thus that $\mathcal{N}$ is irreducible.
    Furthermore $\mathcal{N}$ is finite because both the state spaces and the alphabets are finite.
    Finally since the CRN is execution bounded, one can deduce using Lemma \ref{l.halting} that $\mathcal{N}$ halts on all inputs. This completes the proof.
\end{proof}
For generalized toppling networks,
the kernels of the group actions $\mathbb{Z}^{{\mathcal{A}}_k} \times e_{\nu_k}\mathcal{Q}_{k} \to e_{\nu_k}\mathcal{Q}_{k}$ associated to processor $\mathcal{P}_k$ (see Def. \ref{TotalKernel}) are the additive subgroups of $\mathbb{Z}$ generated by $\Gamma^r_{k,k}$ i.e.:
$$K_{k} =\{\lambda \Gamma^r_{k,k} |\lambda \in \mathbb{Z}\},$$
and the total kernel is
$K = \prod\limits_{1\le k\le p} K_k$. \\

AN theory introduces so-called {\bf production matrices} that permit to give formula for average values of odometers for uniformly distributed recurrent AN states (see the Appendix).  We show below that, for generalized toppling networks, the spectral properties of such production provide information on CRN and AN execution boundedness.

\begin{proposition}\label{p.production}
    The production matrix of an entirely execution bounded generalized toppling network $\mathcal{N}$ does not depend on the choice of the recurrent state $q \in \mathcal{Q}$ and is given by :
    \begin{equation}\label{production_matrix}
     P_{ij} = 
\begin{cases}
\displaystyle  \frac{\Gamma^p_{ij}}{\Gamma^r_{jj}} & \text{if } \Gamma^p_{ij} > 0, \\
0 & \text{otherwise.}
\end{cases} , \quad \forall i,j \in \{1,...,d\}.
\end{equation}
\end{proposition}

\begin{proof}

    Let $j \in \{1,..d\}$. The related production matrix is obtained by considering
    $k\triangleright q   $ for fixed recurrent states $q$ and $k\in K$. The local action $k\triangleright q =\pi_k(k.q)$
    is obtained by adding $k_j$ letters  for each letter type  $ j\in\mathcal{S}$, and by processing each letter once without processing transferred letters, see (\ref{local_action}) from the Appendix. For given $j\in\mathcal{S}$, the effect of processing it 
    $k_j =\Gamma_{jj}^r$ times using the map $\pi_j$ is to provoke  a unique reaction at $j$. Following the algebraic definition of a production matrix (see (\ref{d.prodmatrix}) of the Appendix), we consider the local action of $\Gamma_{jj}^r e_j$ on $q$ to get that
    $$\Gamma_{jj}^r e_j \triangleright q = \sum_{\substack{i=1 \\ i \ne j}}^{N} \Gamma^p_{ij} e_i .q  $$
    Using Definition \ref{d.prodmatrix}, we deduce that
     $P_q(\Gamma_{jj}^r e_j) = \sum\limits_{\substack{i=1 \\ i \ne j}}^{N} \Gamma^p_{ij} e_i$, 
      with
     $\Gamma_{jj}^r e_j \in K$ from construction. 
     We therefore obtain from 
     linearity that $P_q(e_j) = \sum\limits_{\substack{i=1 \\ i \ne j}}^{N}\frac{ \Gamma^p_{ij} }{\Gamma^r_{jj}}e_i$.
     The AN $\mathcal{N}$ being locally irreducible, the production matrix does not depend on the choice of the recurrent state $q \in \mathcal{Q}$.
\end{proof}

\begin{lemma}\label{l.halting}
For generalized toppling networks, the three following assertions are equivalent: 
\begin{enumerate}
    \item  The CRN  $\cH=(\cS,\cC,\cR)$ is entirely execution bounded.
    \item The corresponding Abelian network $\mathcal{N}$ halts on all inputs.
    \item The moduli of the eigenvalues of the associated production matrix $P$ are strictly smaller than one.
\end{enumerate}
\end{lemma}

\begin{proof}
   The equivalence of  statements (2) and (3)  follows from Theorem 5.6 of \cite{bond2016abelian2}.\\
   
    We next prove that $(1) \Rightarrow (3)$.    Let $P$ be the production matrix of $\mathcal{N}$. Let $\lambda$ be its Perron-Frobenius eigenvalue of non-negative eigenvector $x$. Suppose, for contradiction, that $\lambda \geq 1$.
    We follow the proof of Theorem 5.6 in \cite{bond2016abelian2}. There exists a vector $y \in \mathbb{Q}^{\mathcal{A}}$ such that $x \leq y \leq \lambda x$, so that $Py \geq Px = \lambda x \geq y$ which leads to  $Pny \geq Pnx = \la nx \geq ny$ for $n\in\N$. 
    By Lemma 4.5 of \cite{bond2016abelian2}, we can find $n \geq 1$ with $z=ny \in K$,
    that is $z_j = \lambda_j \Gamma_{jj}^r$, $\lambda_j \in\N_+$, $j\in\mathcal{S}$. Hence $z=\sum_j \lambda_j \nu_j$. We next relate the AN and CRN frameworks. First notice that
    $$ (Pz)_i = \sum_j P_{ij}\lambda_j \Gamma_{jj}^r = \sum_j \frac{\Gamma^p_{ij}}{\Gamma^r_{jj}}\lambda_j \Gamma_{jj}^r =\sum_j \Gamma^p_{ij}\lambda_j ,$$
    so that $Pz = \sum_j \lambda_j \nu_j'$, 
    where we have considered the generalized toppling network reactions $\nu_j \to \nu_j'$. The above computations show that
    $$Pz=\sum_j \lambda_j \nu_j' \ge \sum_j \lambda_j \nu_j=z.$$
    Let $\psi$ be a non-negative linear function of the form $\psi(a)=\sum_i v_i a_i$, with $v_i \ge 0$ $\forall i\in\mathcal{S}$. Then
    $\psi(Pz)\ge \psi(z) $, so that $\sum_j \lambda_j (\psi(\nu_j')-\psi(\nu_j))\ge 0$, and there exists $j\in\mathcal{S}$ with
    $\psi(\nu_j')\ge \psi(\nu_j)$, implying that $\psi$ is not a linear potential function (see Definition \ref{d.linear_potential}). The CRN has therefore no linear potential function,
    which is, according to
     Theorem \ref{condition_bounded}, in contradiction with the hypothesis of CRN execution boundedness. Therefore necessarily $\lambda < 1$. \\
 
    We now prove that $(3) \Rightarrow (1)$. The proof is inspired from  the proof of Theorem 6.5 of \cite{doty2024computational}. Suppose that all eigenvalues of $P$ have modulus strictly smaller than $1$.\\
    We first show that there exists no non-negative integer vector $u\in \mathbb{N}^N$ such that $\Gamma u \geq 0$, where $\Gamma$ is the CRN stoichiometric matrix. Assume by sake of contradiction, the existence of such an $u$ with $\Gamma u \geq 0$. Let $z = (z_i)_{1 \leq i \leq N} $ be such that $z_i = \Gamma^r_{ii} u_i$. Then $z \geq 0$, $z \in K$, and 
    \begin{equation*}
        (Pz)_i = \sum_j \Gamma^p_{jj}u_j \geq \sum_j \Gamma^r_{ij}u_j = \Gamma^r_{ii}u_i =(z)_i, \quad \forall i \in \{1,\dots, N\},
    \end{equation*}
    where the inequality is a consequence of the assumption $\Gamma u \geq u$. We deduce from  Collatz-Wielandt formula \cite{collatz} that the Perron-Frobenius eigenvalue of $P$ is larger or equal to $1$, which is a contradiction. So there exists no non-negative integer vector $u\in \mathbb{N}^N$ such that $\Gamma u \geq 0$.\\
     By Corollary 6.4 of \cite{doty2024computational}, there is an integer vector $v \geq 0$ such that $v^T \Gamma  < 0 $. We consider the linear function $\psi: \mathbb{R}^N_{\geq 0} \to \mathbb{R}_{\geq 0}$ given by $\psi(x) = v^Tx$. 
     Using the fact that the reaction vectors $\nu_j'-\nu_j$ correspond to the columns of the stoichiometric matrix, we deduce that
     $\psi(\nu_j')-\psi(\nu_j) <0$, for all $j\in\mathcal{S}$.  $\psi$ is therefore a linear potential function, and the CRN is entirely  execution bounded from Theorem \ref{condition_bounded}.
    
\end{proof}

 Theorem \ref{t.critical} and Theorem \ref{t.number_recurrent_states} from AN theory hold for any finite irreducible AN that halts an all inputs, so that they apply in particular to generalized toppling networks. In this setting, the critical group is isomorphic to $\mathbb{Z}^\mathcal{A} / (I-P)K$ and contains $\mathrm{det}(L)$ elements, where $L$ is the reduced Laplacian, see Definition \ref{reduced Laplacian}.

Considering again the Markov chain of transition mechanism (\ref{next_state_markov}), and an input distribution with adequate support (for instance the uniform distribution on $\mathcal{A}$),  Lemma \ref{l.algebraically_recurrent} and Theorem \ref{t.StationarySand} from the Appendix imply  that the recurrent states of the CRN sandpile Markov chain coincide with the  recurrent states of the AN, and that the stationary distribution is the uniform distribution on $\text{Rec } \mathcal{N}$.

Next Section focuses on the length of AN executions needed to perform a single step
$1_i \triangleright \triangleright q $, $i\in\mathcal{S}$,
when $q$ is uniform on $\text{Rec } \mathcal{N}$
 for dissipative systems.

\subsubsection{Expected time to halt for dissipative generalized toppling networks}\label{ExampleStochasticNew} 
We say that the CRN is \emph{dissipative} when 
  $\sum\limits_i \Gamma^p_{ij}\le  \Gamma_{jj}^r$, for all $j \in \{1,...,d\}$ with a strict inequality for at least some $j$.
  In this  case,

              $$
\sum_{i}P_{ij}=\sum_i \frac{\Gamma^p_{ij}}{\Gamma^r_{jj}}
              =\frac{1}{\Gamma^r_{jj}}\sum_i \Gamma^p_{ij} \le 1,$$
            so that the matrix $P^T$   is sub-stochastic. Hence, the spectral radius of $P$ is  smaller than 1, and strictly smaller than 1 when   $\sum\limits_i \Gamma^p_{ij} <  \Gamma_{jj}^r$ for some $j\in   \{1,...,d\}$.
It follows from Lemma~\ref{l.halting} that the CRN is entirely execution bounded. \\

Consider the enlarged state space
$\mathcal{W}= \mathcal{S}\cup \{\emptyset\}$, and the $(d+1) \times (d+1)$ matrix $S$ with coefficients:
\begin{equation}
    S_{ji}  = 
\begin{cases}
(P^T)_{ji}  & \text{if } i,j \in \{1,...,d\}, \\
1-\sum_{i\in\mathcal{S}}S_{ji} & \text{if } j \in \{1,...,d\}  \text{ and } i= \emptyset,\\
0 & \text{otherwise},
\end{cases}
\end{equation}

From construction, the restriction of $S$ on $\mathcal{S}$ is just $P^T$.
The matrix $S$ is thus stochastic and defines a Markov chain $(Y_n)_{n\ge 0}$ of state space
$\mathcal{W}$, and absorbing state $\emptyset$. The related fundamental matrix $F$  with indexes from $\mathcal{S}$
which describes non-absorption in $\{\emptyset\}$ is $F=(1-P^T)^{-1}$, with
$$F = \sum_{k\ge 0}(P^T)^k,$$
so that, for $j\in\mathcal{S}$, 
\begin{eqnarray*}
F_{ji} &=&\sum_{k\ge 1} (P^T)^k_{ji} =\sum_{k\ge 1}  (S^k)_{ji} =
    \sum_{k\ge 1}  \mathbb{P}(Y_l\ne \emptyset,\ l\le k;\ Y_k= i\vert Y_0 =j)\\
    &=&  \sum_{k\ge 1}  \mathbb{P}(\tau > k;\ Y_k= i\vert Y_0 =j),\ j\ne i,\\
\end{eqnarray*}
where $\tau=\min\{n\ge 1;\ Y_n =\emptyset\}$  is the absorption time.\\
Recall that the odometer $[x.q]$ is 
$\vert w\vert$ for a complete and legal execution $w$ of $x.q$, that is, it gives the number of letters of each type processed during the execution
of $w$.
Theorem \ref{expected_numbers_letters}
gives that the mean value of the odometer with respect to uniform random elements $q$ from $Rec\ \mathcal{N}$ is:
$$\mathbb{E}[x.q]^T = x^T (I - P^T)^{-1} = x^T F.$$
\\

Set $\vert\vert x\vert\vert =\sum\limits_i x_i$, and consider the probability measure $\nu_x$ on $\mathcal{S}$ defined by $\nu_x(i)=x_i/\vert\vert x\vert\vert$. 
Standard Markov chain theory \cite{grinstead2006grinstead} shows then that the average total number of processed letters  is:
\begin{equation}\label{AverageMove}
\sum_i \mathbb{E}[x.q]^T_i  =  \vert\vert x\vert\vert\  \nu_x F (1,\ldots,1)^T = \vert\vert x\vert\vert\
\E_{\nu_x}(\tau),
\end{equation}
where $\nu_x$ is considered as a line probability vector.
Hence, the total number of iterations to move the random state $x.q$ for a uniform $q$ of 
$Rec\ \mathcal{N}$ toward the state $x \triangleright \triangleright q$ is the total number of letters to be processed $\vert\vert x\vert\vert$
multiplied by the mean absorption time $\E_{\nu_x}(\tau)$, when $Y_0$ is distributed according to the probability measure $\nu_x$.\\

Consider the AN sandpile Markov chain  of transition mechanism
$q_{n+1} = 1_{i_n} \triangleright \triangleright q_n$.
  Equation (\ref{AverageMove}) shows that
the typical AN execution length  needed for computing a step $q_n\to q_{n+1}$ of the stationary chain is given by $\E_{\delta_{i_n}}(\tau)$, where
$\delta_{i_n}$ is the point measure centered at $i_n\in\mathcal{S}$.

\subsection{A special case of multi-species source complexes\label{s.multiple_species}} 

As discussed earlier, a general problem with multi-species source complexes, is the accumulation of molecules of some species. This can lead to an infinite state space, and the previous theory of finite Abelian networks does not apply anymore.\\ In this subsection, we avoid this problem by considering only complexes of type 1a (see Definition \ref{ComplexTypes}), and a slightly modified Markov chain, called the batch version of the sandpile Markov chain, see  (\ref{Batch}). For this Markov chain to solve the problem of infinite state space, we need all complexes to have the same number of molecules of each source species i.e. for all $k, k' \in \{1,...,d\}$, $\nu_k = \beta_k e_{U_k}$ and  $\nu'_{k'} = \beta'_{k'} e_{U_{k'}}$ for some choice of $\beta_k$ and $\beta'_{k'}$. Any reaction is thus of the form:
\begin{equation}
\beta_k e_{U_k} \to \beta'_{k'} e_{U_{k'}} 
\end{equation}

As stated before there is a single reaction having $e_{U_k}$ as a source complex.
\subsubsection{Batch version of the sandpile Markov chain}
One natural objective is to visit the static space as proposed in
\cite{chen2023rate} and \cite{Vasic} using a reasonable stochastic process. This can be obtained by using the following
transition mechanism:
\begin{equation}\label{Batch}
q_{n+1} = 1_{U_{i_n}} \triangleright \triangleright q_n,
\end{equation}
where the sequence $U_{i_n}$ is i.i.d. in $\{U_1,...,U_p\}.$ Set for convenience (see (\ref{BasicMap})):
\begin{equation}\label{BasicMap2}
\tau_U (q) = 1_U \triangleright \triangleright q.
\end{equation}
Lemma 2.2 of \cite{bond2016abelian3} shows that for all $k\in \{1,...,p\}$:
$$\tau_{U_k} = \prod_{j\in U_{k}}\tau_j.$$
We associate to each source complex a processor $\mathcal{P}_{U_k}$. The alphabet of this processor is $\mathcal{A}_k = \{U_k\}$.
The state space of the processor depends on the initial state $\eta \in \mathbb{N}^N$, we indeed have that:
$$\mathcal{Q}^k(\eta) = \{\eta -\chi_k e_{U_k},...,\eta - e_{U_k},\eta, \eta + e_{U_k},...,\eta + \epsilon_k e_{U_k}\}$$ where $\chi_k = \min\limits_j \{\eta_j  \}$ and $\epsilon_k = \max\limits_j\{\Gamma^r_{j,k}-\eta_j \} $. \\
Note that two different initial conditions lead to the same state space if and only if they are in the same discrete stoichiometric compatibility class.\\
For each $U_k$, there are $\# U_k \times (\beta_k -1)   $ distinct discrete stoichiometric compatibility classes since the relevant quantity is the net balance of molecules within $U_k$. Consequently, the total number of different discrete stoichiometric compatibility classes is given by:  
\begin{equation}\label{nb_compatibility_classes}
    \prod\limits_ {{k=1,\cdots, p}} \# U_k\times(\beta_k -1)
\end{equation} 
We define the transition and message passing functions similarly as in the single-species case (\ref{TransitionNew},\ref{MessagePassing}) but letters are now constituted of one molecule of each species of $U_k$. Once the initial state is fixed, the local action is irreducible and the associated Markov chain of transition mechanism given by (\ref{Batch}) possesses a unique recurrent class, and thus one can prove that:
\begin{proposition}
    Given an initial condition, the collection of processors $\mathcal{N}=(\mathcal{P}_{U_k{}})_{k=1,\cdots, p}$ on the directed graph ${\mathcal G}=({\mathcal C},{\mathcal E})$ forms a finite irreducible Abelian network that halts on all inputs.
\end{proposition}
Therefore once an initial condition has been fixed, all the previous theoretical results concerning group structure and recurrent states, can be obtained in a similar manner as before (see Example \ref{multispecies_example}).
In particular we can follow \cite{meester2003abelian}, and derive some interesting properties of the batch version of the  sandpile Markov chain. \\
 Let $\mathcal{R}_e$ be the set of recurrent states of the Markov chain and $\mathcal{M}_s$ be the subset of the global monoid $\mathcal{M}$ given by:
$$\mathcal{M}_s=\{\prod_k \tau_{U_k}^{n_k},\ n_k\in \N\},$$
acting on the set of static states, which is an abelian sub-monoid of $\mathcal{M}$.

\begin{lemma}\label{SubMonoid}
The following assertions are extensions of the classical results for the ASM (Proposition 3.1 of \cite{meester2003abelian}):
\begin{itemize}
 \item{} $\mathcal{R}_e$   is closed under the action of $\mathcal{M}_s$.
 \item{} Let $\eta\in\mathcal{R}_e$. There exists $n(\eta)=(n_1(\eta), ...,n_N(\eta)) \in \N^N$ such that $\prod_k \tau_{U_k}^{n_k(\eta)} \eta =\eta$,
 \item{} Let $\eta\in\mathcal{R}_e$, and let $\mathcal{R}_e(\eta)$ be the recurrent class containing $\eta$. Let
 $$A = \{\xi \in \mathcal{R}_e;\ \prod_k \tau_{U_k}^{n_k(\eta)}\xi =\xi\}.$$
 Then $A = \mathcal{R}_e(\eta)$.
 \item{} Let $\eta\in\mathcal{R}_e$. The restriction of $\mathcal{M}_s$ to $\mathcal{R}_e(\eta)$  is an abelian group of neutral element
 $$e(\eta)=\prod_k  \tau_{U_k}^{n_k(\eta)},$$
 \item{} The inverse of any $\tau_{U_k} $ is
 $$\tau_{U_k}^{-1}= \tau_{U_k}^{n_k(\eta)-1}\prod_{l\ne k}\tau_{U_l}^{n_l(\eta)}.$$
 \item{} Assume that $\Gamma_{jk}^p  \equiv 0$, when $j\in U_k$, $\forall k$, and that $\Gamma_{.,k}^r=\gamma_k 1_{U_k}$ for constants $\gamma_k\in\N$,   $\forall k$. The following closure relation holds true: for all k
 \begin{equation}\label{Closure}
 e=\tau_{U_k}^{\gamma_k}\prod_{i\not\in U_k}\tau_i^{-\Gamma_{ik}^p}=\prod_{j\in U_k}\tau_j^{\Gamma_{jk}^r}\prod_{i\not\in U_k}\tau_i^{-\Gamma_{ik}^p},
 \end{equation}
 acts as neutral element for $\mathcal{M}_s$ acting on $\mathcal{R}_e$. Given $\eta\in\mathcal{R}_e$, the restriction of $e$ on   $\mathcal{R}_e(\eta)$ is $e(\eta)$.
\end{itemize}
\end{lemma}

\begin{proof}
Assume that $\eta\in\mathcal{R}_e$, and let $g\in\mathcal{M}_s$. Noticing that the Markov chain moves are of  the form
$\eta \to \tau_{U_k}\eta$ for some $k$, we see that $g\eta \in \mathcal{R}_e$, so that $\mathcal{R}_e$ is closed under the action of $\mathcal{M}_s$.
Concerning the second assertion, $\eta\in\mathcal{R}_e$ implies that there is a Markov chain path, that is a composition of operators
$\tau_{U_k}$ taking $\eta$ back to itself. For the third point, $A\ne\emptyset$ since $\eta\in A$.
Suppose that $\xi\in A$, and consider the action of some $\prod_k \tau_{U_k}^{m_k}\in \mathcal{M}_s$ on $\xi$: 
$$
\prod_k \tau_{U_k}^{m_k}\xi = \prod_k \tau_{U_k}^{m_k} \prod_k \tau_{U_k}^{n_k(\eta)} \xi 
                            = \prod_k \tau_{U_k}^{n_k(\eta)}\prod_k \tau_{U_k}^{m_k}\xi,$$
showing that $\prod_k \tau_{U_k}^{m_k}\xi\in A$.
$A$ is therefore closed under the action of $\mathcal{M}_s$, and also for the action of the Markov chain. Hence $A$ contains $\mathcal{R}_e(\eta)$
and therefore $A = \mathcal{R}_e(\eta)$.
The preceding argument also shows that
$$e(\eta)= \prod_k \tau_{U_k}^{n_k(\eta)}$$
acts on $\mathcal{R}_e(\eta)$ as a neutral element. The inverse of $\tau_{U_k}$ is obtained from direct computation using commutativity. 
Concerning the last point,
 Lemma 2.3  of 
\cite{bond2016abelian3} states that if $x \triangleright q = y.r$, then $x\triangleright\triangleright q = y\triangleright\triangleright r  $, so that
$$\tau_{U_k}^{\gamma_k}=\prod_{j\in U_k}\tau_j^{\Gamma_{jk}^r} =\prod_{i\not\in U_k}\tau_i^{\Gamma_{ik}^p}\in \mathcal{M}_s,$$
giving the result.
\end{proof}

\section{Examples of noncompetitive CRNs}\label{s.examples}

\begin{example}\label{exemple1}
     The following noncompetitive CRN, taken from \cite{Vasic}, implement the simplest dissipative ASM on a graph described by three nodes $1,2$ and $3$
     on a circle, where $3$ models the open boundary where sand can be evacuated
     (see Section~\ref{s.SandB}):
          \begin{align*}
    2e_1 &\longrightarrow e_2 + e_3, \\
    2e_2 &\longrightarrow e_1 +e_3,
\end{align*}

In this example, $U_0=\{3\}$, $U_1=\{1\}$  and $U_2=\{2\}$. All complexes are of type 1a.
Following Remark~\ref{r.sink}, species $3$ will be considered as the sink since it only appears as a product and does not participate to the dynamical evolution of the abundances of 
species $1$ and $2$.
We therefore focus on the following noncompetitive CRN:

\begin{align*}
    2e_1 &\longrightarrow e_2 , \\
    2e_2 &\longrightarrow e_1,
\end{align*}
so that $U_0 =\emptyset$, $U_1=\{1\}$  and $U_2=\{2\}$, with all complexes of type 1a. This CRN is a  generalized toppling network that
coincides with the introductory  Example \ref{e.intro}.  
The static states are $q_1 = (0,0)^T,q_2 = (0,1)^T,q_3 = (1,0)^T$ and $q_4 =(1,1)^T$. 
This CRN is dissipative (see Section~\ref{ExampleStochasticNew}), and following (\ref{production_matrix}), the sub-stochastic production matrix is given by:
$$P = \begin{pmatrix}
    0 & 1/2 \\
    1/2 & 0
\end{pmatrix}.$$

The spectral radius being strictly smaller than 1, Lemma~\ref{l.halting} ensures that this CRN is execution bounded. The kernel of the group action is $K = 2\mathbb{Z}\times2\mathbb{Z}$ (\ref{e.kernel}).\\
The reduced Laplacian (see Definition \ref{reduced Laplacian}) is : 
$$L = \begin{pmatrix}
    2 & -1 \\
    -1 & 2
\end{pmatrix}$$ with determinant $det(L) =3$. By Theorem \ref{t.number_recurrent_states}, there are thus $3$ recurrent states and $3$ elements in the abelian group $\text{Crit } \mathcal{N}$.\\
We now follow the burning algorithm for ANs (see Algorithm \ref{alg:burning}). We first find the burning element $y = (1,1)^T$, and we obtain that $k = Dy =(2,2)^T \in K$ which clearly satisfies $k\geq 1$ and $Pk\leq k$. We then apply the second step of the algorithm to find the recurrent states.\\
One can check that for $i \in \{2,3,4\}$, we have $(I-P)k \triangleright \triangleright q_i = q_i$, whereas this does not hold for $q_1$. Using Theorem \ref{t.burning_algo}, this proves that the recurrent states are $q_2,q_3$ and $q_4$ (see Figure \ref{f.intro_example}). The sandpile Markov chain of Section \ref{s.SandMC-AN} with a uniform input distribution, will thus converge toward a uniform distribution on the recurrent states $q_2$, $q_3$ and $q_4$, according to Theorem \ref{t.StationarySand}.
The elements of the critical group are the maps $\tau_1$, $\tau_2$ and $\tau_1 \tau_2$ (see (\ref{BasicMap})) and are in one-to-one correspondence with recurrent states.

Let $q$ be a random recurrent element, and choose for instance $x = (3,3)^T$. Then following Theorem \ref{expected_numbers_letters}:  $$\mathbb{E}[x.q] = (I-P)^{-1}x = \begin{pmatrix}
    4/3 & 2/3 \\
    2/3 & 4/3
\end{pmatrix} \begin{pmatrix}
    3 \\ 3
\end{pmatrix} = \begin{pmatrix}
    6 \\ 6
\end{pmatrix}$$
 meaning that $6$ molecules of each species will be processed in average if the AN starts at state $x.q$.

\end{example}
\begin{example}\label{second_example}
        Consider another noncompetitive CRN which describes a generalized toppling network:
        \begin{align*}
    2e_1 &\longrightarrow 2e_2+e_3, \\
    2e_2 &\longrightarrow e_1 \\
    3 e_3 &\longrightarrow e_2
\end{align*}

There is no sink and all complexes are of type 1a with $U_0 =\emptyset$, $U_1=\{1\}$, $U_2=\{2\}$, and $U_3=\{3\}$. We apply the methodology developed in Section~\ref{s.single_species}. \\ The production matrix is:
$$P = \begin{pmatrix}
    0 & \frac{1}{2} & 0 \\
     1 & 0 & \frac{1}{3} \\
      \frac{1}{2} & 0 & 0 
\end{pmatrix}$$
The spectral radius being strictly smaller than 1, Lemma~\ref{l.halting} gives that  this CRN is execution bounded.  The kernel of the group action is $K = 2\mathbb{Z}\times2\mathbb{Z}\times3\mathbb{Z}$.\\
There are $12$ static states and as before our goal is to determine the recurrent states, since the stationary distribution sandpile Markov chain will only be concentrated at these states. \\
The reduced Laplacian of the system is:
$$L = \begin{pmatrix}
    2 & -1 & 0 \\
     -2 & 2 & -1 \\
      -1 & 0 & 3 
\end{pmatrix}$$ with determinant equals to $5$ indicating that there are $5$ recurrent states among the $12$ states that are static. Equivalently there are $5$ elements in the critical group, since there is a bijection between the critical group $\text{Crit } \mathcal{N}$ and the set of recurrent states $\text{Rec }\mathcal{N}$. 
To find these recurrent states, we follow the burning algorithm (Algorithm \ref{alg:burning}). We first find that $y = (1,2,1)^T $ is a burning element and it follows that $k = Dy = (2,4,3)^T \in K$ and, satisfies $k\geq 1$ and $Pk\leq k$. Then we apply the second step of the burning algorithm, and we check that $q_1 = (1,1,0)^T$, $q_2 = (1,0,1)^T$, $q_3 = (1,1,1)^T$, $q_4 = (1,0,2)^T$ and $q_5 = (1,1,2)^T$ are the only static states that verify $(I-P)k \triangleright \triangleright q_i = q_i$, and therefore: $$\text{Rec }\mathcal{N} = \{q_1,q_2,q_3,q_4,q_5\}$$
Figure \ref{f.second_example} illustrates, that once the sandpile Markov chain enters the set of recurrent states $\text{Rec }\mathcal{N}$, it will stay in it forever.
Furthermore with a uniform input distribution, this Markov chain will converges toward a uniform distribution on $\text{Rec }\mathcal{N}$.
\begin{figure}[h]\label{f.second_example}
    \centering
    \includegraphics [width=12cm]{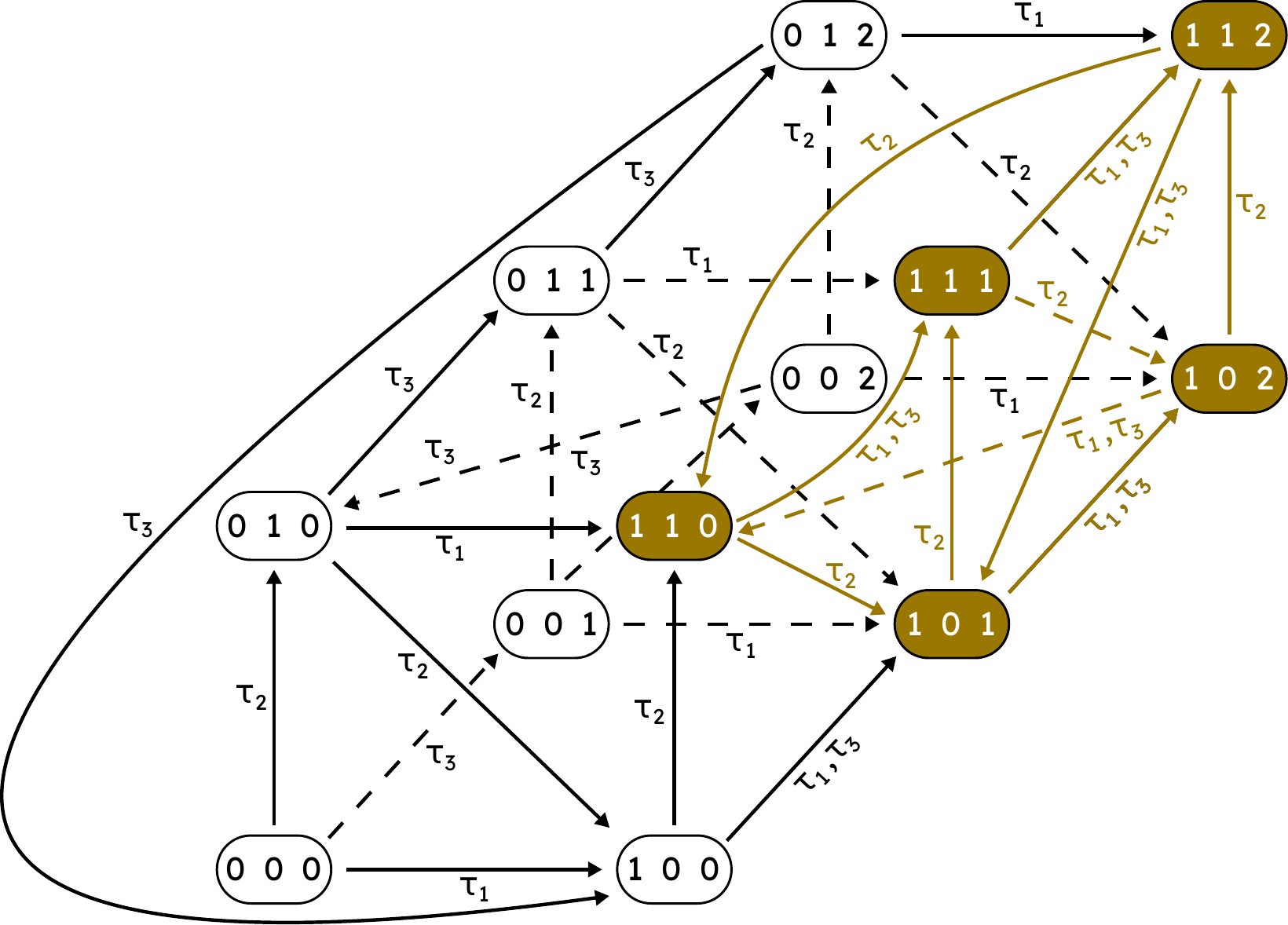}
    \caption{
    The transition diagram of the sandpile Markov chain associated with the noncompetitive CRN of reactions
    $2e_1 \longrightarrow 2e_2+e_3$,
    $2e_2 \longrightarrow e_1$ and
    $3 e_3 \longrightarrow e_2$ of Example \ref{second_example}. 
    We show the results of  critical group generators (\ref{BasicMap}) on transitions. The five recurrent static states among the twelve static states are shaded brown.
    }
\end{figure}
\end{example}

\begin{example}\label{multispecies_example}
    Now introduce an example of a noncompetitive CRN falling in the range of Section \ref{s.multiple_species} :
        \begin{align*}
    &2e_1 \longrightarrow e_2+e_3, \\
    &2e_2+2e_3 \longrightarrow e_1 
\end{align*}

We remark that $U_0 = \emptyset$ , $U_1 = \{1\}$ and $U_2=\{2,3\}$. All complexes are of type 1a and writes $\nu_k = \beta_k e_{U_i}$. We can then write : 
    \begin{align*}
    &2e_{U_1} \longrightarrow e_{U_2} \\
    &2e_{U_2} \longrightarrow e_{U_1} 
\end{align*}
We retrieve Example \ref{exemple1}, where for all $k \in \{1,2\}$, $e_{k}$ is replaced by $e_{U_k}$. The production matrix, the kernel or the reduced Laplacian matrix are given in Example \ref{exemple1}. 
Using (\ref{nb_compatibility_classes}), we find that there are $6$ distinct stoichiometric compatibility classes for which the the state space and thus the recurrent states will be different. \\
Let note $\eta_i$ for $i \in \{1,...,6\}$, be an initial state associated to the discrete stoichiometric compatibility class $\mathcal T + \eta_i$ . We choose $\eta_1 = (0,0,0)^T$, $\eta_2 = (0,1,0)^T$, $\eta_3 = (0,0,1)^T$, $\eta_4 = (1,0,0)^T$, $\eta_5 = (1,1,0)^T$ and $\eta_6 = (1,0,1)^T$. These states don't communicate for the batch Markov chain given in (\ref{Batch}), and therefore choosing them as initial states leads to different state spaces. 
We follow Example \ref{exemple1}, and deduce that for all $i \in \{1,...,6\}$ the states $\eta_i+ e_{U_1}$, $\eta_i+ e_{U_2}$ and $\eta_i+ e_{U_1}+ e_{U_2}$ are recurrent but not $\eta_i$.

\end{example}

\begin{example}\label{ExempleTopplingNN}

Consider the noncompetitive CRN associated with RReLU NN of Section \ref{ReLUCRN}, where we assume for simplicity that the weights $w_k =p_k/l_k$, $k=1,\ldots,d$
are positive. The species set is $\mathcal{S}=\{1,\ldots,d+1\}$, $N=d+1$,
corresponding to the species symbols $X_1^+,\ldots,X_d^+$ and $Y^+$,
and the set of single-species reactions is given by
$$
        l_k X_k^+ \longrightarrow p_k Y^+,
\ k=1,\ldots, d.$$
 We just focus on the linear part of the NN omitting the ReLU function.
Note that $U_0 =\{Y^+\}$ for the sink species $Y^+$, so that the main assumption
$U_0 =\emptyset $ of Section \ref{s.single_species} is not satisfied.
We define the processor state spaces  to be
$$\mathcal{Q}_k =\{\gamma_{kk}e_k;\ \gamma_{kk}\in\N,\ \gamma_{kk}< \Gamma_{kk}^r = l_k\},$$
$k=1,\ldots,d$, and,  set $\mathcal{Q}_N = \{\gamma_{NN}e_N;\ \gamma_{NN}\in\N\}$, for processor $\mathcal{P}_N$ associated to species $N$.
The global state space $\mathcal{Q}=\mathop{\oplus}\limits_{k=1,\ldots, d+1} \mathcal{Q}_k$ is therefore infinite.
The transition functions are given as (\ref{TransitionNew}) when $k\le d$ and set
$T_N(q^N)= q^N +e_N$, $q^N\in\mathcal{Q}_N$. The message passing functions $T_{(k,N)}$ from processor
$\mathcal{P}_k$ to processor $\mathcal{P}_N$ is 
\begin{equation*}
T_{(k,N)}(q^k) = \left\{
    \begin{array}{ll}
    N^{p_k}  & \text{if}  \quad q^k_k =l_k -1,\\
          \epsilon  & \mbox{otherwise},
          \end{array}.\right. 
\end{equation*}
$k=1,\ldots,d$.\\

Recall from Section \ref{ReLUCRN} that each term  $w_k x_k^+(0)$ of  the NN weighted sum corresponds to an increase of
$[x_k^+(0)/l_k]p_k$ units from CRN species  $Y^+$, that are considered as transferred letters through the message passing function $T_{(k,N)}$. The total number of processed letters in the AN version of the CRN might thus corresponds to the NN output sum $W=\sum_{k=1}^d w_k x_k^+(0)$. One might thus guess that the $N$th component of the  odometer $[x(0)^+ .q]$ is related to $W$.
Taking inspiration from the results of Proposition \ref{p.production},  consider the matrix
$P_{kN}=0$, $k=1,\ldots,d+1$, $P_{kj}=0$, $k,j=1,\ldots,d$, $P_{Nk}=p_k/l_k$, $k=1,\ldots,d$. This matrix is such that
$(P^T)^n \equiv 0$, $n\ge 2$, so that from Neumann series, $(I-P^T)^{-1}=I +P^T$. The formula for the mean odometer given in Theorem \ref{expected_numbers_letters} 
suggests to consider
$$(x_1^+(0),x_2^+(0),\ldots,x_d^+(0),0)(I-P^T)^{-1}_N = \sum_{k=1}^d w_k x_k^+(0)=W,$$
which leads to a correct result even if 
the  basic assumptions of Proposition \ref{p.production} and Theorem \ref{expected_numbers_letters} are not satisfied.

\end{example}

\section{Discussion}\label{s.discussion}
We have established that rate-independent noncompetitive CRNs can be associated with Abelian Networks.This connection brings to light several structural and dynamical similarities between these two frameworks, and  allows the use of  powerful algebraic tools developed for finite ANs—such as critical groups, production matrices, and burning algorithms— for studying execution bounded noncompetitive CRNs. 
However, this theory only applies when the state space of the CRN is finite, which is no longer the case with multi-species source complexes, since nothing prevents the Markov chain from accumulating arbitrarily large quantities of certain molecular species.  This phenomenon of unbounded accumulation also appears in non-unary ANs such as, e.g., in the oil-water model considered in   \cite{candellero2017oil},
which is concerned with
 scaling-limit and aggregation behaviors rather than algebraic invariants, and thus addresses a substantially different set of questions than those considered here. Developing these extensions represents an exciting direction for future research.

\section{Appendix}\label{BasicsAbelianNetwork}

 We give more details concerning the algebraic aspects of AN that were briefly described in Section \ref{s.AbelianNet}

For a set $\mathcal{Y}$, we define $\text{End}(\mathcal{Y})$ to be the monoid of all set maps $\mathcal{Y} \to \mathcal{Y}$, with the composition operation.\\
Given $x\in \mathbb{N}^{\mathcal{A}}$,  let $\pi_x$ denote $\pi_w$ for a word $w$ such that $\vert w\vert = x$ with $\pi_x(x.q)$. Then 
\begin{equation}\label{local_action}
x \triangleright q = \pi_x (x.q)\in \mathbb{N}^{\mathcal{A}}\ {\rm x}\ \mathcal{Q},
\end{equation}
 is called the \textit{local action}, since each processor processes only letters added at its own vertex, not those passed from other processors. This distinguishes it from the \textit{global action} defined as: 
\begin{equation}x \triangleright \triangleright q = q'\end{equation}
which denotes the final processor state obtained by processing the word $x$, that is, $\pi_w(x.q)= 0.q'$ for a complete and legal execution $w$.

For $a \in \mathcal{A}$, consider the map
 $\tau_a: \mathcal{Q} \to \mathcal{Q}$:
 \begin{equation}
\tau_a (q) = 1_a \triangleright \triangleright q. 
\end{equation}
The set $M_v = \left\langle t_a \right\rangle_{a \in \mathcal{A}_v}$ is a submonoid of $\text{End}(\mathcal{Q}_v)$ called the transition monoid of processor $ \mathcal{P}_v$ and the set $M=\left\langle\tau_a\right\rangle_{a \in \mathcal{A}} \subset \text{End}(\mathcal{Q})$ is a global monoid. 
\begin{definition}\label{d.finite_AN}
    An abelian network $\mathcal{N}=(\mathcal{P}_v)_{v \in V}$ is finite if both $V$ is finite and if each processor $\mathcal{P}_v$ has finite alphabet $\mathcal{A}_v$ and state space $\mathcal{Q}_v$.
\end{definition}
\begin{definition}\label{d.irreducibility}
    An Abelian network $\mathcal{N}=(\mathcal{P}_v)_{v \in V}$ is locally irreducible, if each processor is irreducible i.e. if the associated monoid action $ \tau : M_v \times \mathcal{Q}_v \rightarrow \mathcal{Q}_v$ admits a unique equivalence class (see \cite{bond2016abelian2}, section 2.1). Levine et al. proved that local irreducibility implies global irreducibility \cite{bond2016abelian3}.
\end{definition}
The \textit{minimal idempotent} of $M$ is denoted $e$ and is the unique idempotent accessible from every element of $M$ (i.e. the unique $e \in M$ such that $ee=e$ and $e \in Mm $ for all $m \in M$ ).  For a distribution $\mu$ on $\mathcal{A}$, we denote $M_\mu$ the submonoid of 
$M$ generated by $\{\tau_a: \mu(a)>0\}$.

\begin{definition}\label{CriticalGroup}
    We denote by $\text{Crit } \mathcal{N}$   the \textit{critical group} of an Abelian network $\mathcal{N}$. This is the group $eM$.
\end{definition}
\begin{definition} \label{d.alg_recurrent}
We denote by $\text{Rec } \mathcal{N}$ the set of recurrent states of $\mathcal{N}$, where a state $q \in \mathcal{Q}$ is  recurrent if it satisfies $eq = q$ for the group action $M \times Q \to Q$.
\end{definition}
\begin{proposition}
    For a finite irreducible Abelian networks $\mathcal{N}$ that halts on all inputs, the action of $\text{Crit } \mathcal{N}$ on $\text{Rec } \mathcal{N}$ is a bijection, and thus :
    $$\#\text{Crit } \mathcal{N} = \#\text{Rec } \mathcal{N}$$
\end{proposition}
We can similarly define $e_v$ as the minimal idempotent of the local monoid $M_v$ and say that a state $q\in \mathcal{Q}$ is \textit{locally recurrent} if $q_v \in e_v\mathcal{Q}_v$ for all $v \in V$. \\
We now define the \textit{total kernel} and the \textit{production matrix}, which will be used to describe the critical group.
\begin{definition}\label{TotalKernel}
    The total kernel is defined as:
$$K = \prod_{v \in V} K_v \subset \mathbb{Z}^{\mathcal{A}},$$

where $K_v$ is the kernel of the group action $\mathbb{Z}^{\mathcal{A}_v} \times e_v \mathcal{Q}_v \to e_v \mathcal{Q}_v.$
\end{definition}
\begin{definition}\label{d.prodmatrix}
Given a state $q \in \mathcal{Q}$, the \textit{production map} $P_q$ is the linear map $P_q: K \cap \mathbb{N}^{\mathcal{A}} \to \mathbb{N}^{\mathcal{A}}$ defined by $k\triangleright q = P_q(k).q$. This map extends to a group homomorphism $K \to \mathbb{Z}^{\mathcal{A}}$, which does not depend on $q$, if the Abelian network is locally irreducible.
The \textit{production matrix} $P$ is the matrix of this group homomorphism.
\end{definition}
Intuitively, the production matrix of a finite AN $\mathcal{N}$ is a matrix whose entry $P_{ab}$ represents the average number of letters $a$ produced when processing the letter $b$.
\begin{theorem}\label{t.critical}[\cite{bond2016abelian3},Theorem 3.11]
The critical group then has the algebraic representation:
\begin{equation}\label{critical_group_2}
    \text{Crit }\mathcal{N} \simeq \mathbb{Z}^\mathcal{A} / (I-P)K
\end{equation}
\end{theorem}
Now define the reduced Laplacian matrix which will be useful to count the number of elements in $ \text{Crit }\mathcal{N}$ or equivalently the number of recurrent states.
\begin{definition}\label{reduced Laplacian}
    The reduced Laplacian matrix of a finite irreducible Abelian networks is defined as 
    $$L = (I-P)D$$
    where $D$ is the diagonal matrix with diagonal elements $d_a = \min_{n}\{n1_a \in K\}$.
\end{definition}

\begin{theorem}\label{t.number_recurrent_states}[\cite{bond2016abelian3},Theorem 3.18]
    \begin{equation*}
        \#\text{Crit } \mathcal{N} = \#\text{Rec } \mathcal{N} = \frac{det(L)}{\iota}
    \end{equation*}
    where $\iota$ is the index of $D\mathbb{Z}^{\mathcal{A}}$ as a subgroup of $K$.
\end{theorem}
In the following, we state that these  recurrent states, are exactly the recurrent states of a well-chosen Markov chain, and give an algorithm to find them.

\subsection{The Sandpile Type Markov chain for Abelian networks.}

We follow \cite{bond2016abelian3,meester2003abelian} by considering a Markov chain $(q_n)_{n\geq 1}$, $q_n \in \mathcal{Q}$,
 of
 transition mechanism 
$$
        q_{n+1} = 1_{a_n} \triangleright \triangleright q_n = \tau_{a_n}(q_n),
        $$
where the sequence $(a_n)_{n\ge 1}$ is i.i.d. from a distribution $\mu$ on $\mathcal{A}$.
Intuitively, a letter $a_n$ is chosen at random in $\mathcal{A}$ and the unique processor $\mathcal{P}_v$ with $a_n \in \mathcal{A}_v$ processes this letter to arrive at the new state
$1_{a_n} \triangleright \triangleright q_n$.

\begin{proposition}\label{p.accessibility}
A state $q$ is recurrent for the Markov chain (\ref{next_state_markov}) if and only if $q$ is accessible from $mq$ for all $m \in M_{\mu}$.
\end{proposition}

    Following the theory of Abelian sandpile \cite{meester2003abelian} and Abelian networks \cite{bond2016abelian1,bond2016abelian2,bond2016abelian3}, we derive some nice properties, in particular that the stationary distribution of this Markov chain is uniform  on the set of recurrent states.
    
    \begin{definition}\label{d.adequate-support}
       A distribution $\mu$ on $\mathcal{A}$ is said to have adequate support if $e_\mu M_\mu = eM$ .
    \end{definition}
The following lemma gives a one-to-one correspondence between recurrent states of the sandpile Markov chain and Rec $\mathcal{N}$.
\begin{lemma} (\cite{bond2016abelian3}, Lemma 3.5) \label{l.algebraically_recurrent}
    If $\mu$ has adequate support, the set of recurrent states  Rec $\mathcal{N}$  coincide with the set of recurrent states for the sandpile Markov chain (\ref{next_state_markov}).
\end{lemma}

The stationary distribution of this Markov chain is analogous to the one of the Abelian sandpile model \cite{meester2003abelian}.
\begin{theorem}\label{t.StationarySand}(\cite{bond2016abelian3}, Theorem 3.6)
    If $\mu$ has adequate support, the uniform distribution on $\text{Rec }\mathcal{N}$ is the unique stationary distribution of the sandpile Markov chain.
\end{theorem}

To find concretely which static states are recurrent   
Levine \textit{et al.} \cite{bond2016abelian3} generalized Dhar's burning algorithm to any finite irreducible Abelian network $\mathcal{N}$ that halts on all inputs.
\begin{definition}\label{d.burning_element}[\cite{bond2016abelian3}, Definition 5.1]
    A \textit{burning element} of an Abelian network $\mathcal{N}$ is a vector $\beta \in \mathbb{N}^{A}$ such that for $q \in \mathcal{Q}$ : 
    $$q \in \text{Rec } \mathcal{N} \Leftrightarrow \beta \triangleright \triangleright q = q$$
\end{definition}
Such a burning element always exists for a 
for a finite irreducible Abelian network $\mathcal{N}$ that halts on all inputs. The following procedure allows to find one.
\begin{lemma}\label{l.procedure}[\cite{bond2016abelian3}, Procedure 5.8]
    To find a burning element, we start with $y=1$. If $Ly\geq 0$, then we stop. Otherwise we choose some $a \in \mathcal{A}$ such that $(Ly)_a < 0$ and increase $y_a$ by $1$. Repeat until $Ly \geq 0$.
\end{lemma}
Once we have a burning element, one can easiliy verify which static states are recurrent.
\begin{theorem}\label{t.burning_algo}[\cite{bond2016abelian3}, Theorem 5.5]
  Let $k \in K$ be such that $k \geq 1$ and $Pk\leq k$. Then $q \in \mathcal{Q}$ is recurrent if and only if $(I-P)k \triangleright \triangleright q = q.$
\end{theorem}

Finally, a result concerning the expected number of letters processed before the network halts, is derived:
\begin{theorem}(\cite{bond2016abelian3}, Theorem 3.7)\label{expected_numbers_letters}
    Let $q$ be a uniform random element of $\text{Rec }\mathcal{N}$.
Then for all $x \in \mathbb{N}^{\mathcal{A}}$, we have
$$\mathbb{E}[x.q] = (I - P)^{-1}x,$$
where $\mathbb{E}$ denotes the mathematical expectation with respect to the uniform measure on $\text{Rec }\mathcal{N}$.
\end{theorem}

 \bibliographystyle{plain}

 \bibliography{biblio} 
\end{document}